\RequirePackage[2014/01/01]{latexrelease}
\documentclass[journal]{IEEEtran}
\usepackage{graphicx,amsmath,amsfonts,amssymb,mathrsfs,cite,epstopdf,amsthm,url,xcolor} 
\usepackage{stfloats}
\usepackage{subcaption,rotating}
\usepackage[hidelinks]{hyperref}
\hypersetup{
    colorlinks=true,       
    linkcolor=blue,          
    citecolor=green,        
    urlcolor=blue           
}

\usepackage[normalem]{ulem} 

\def\d{\mathsf{d}}
\def\nb{\mathsf{n_b}}
\def\nd{\mathsf{n_d}}
\def\nr{\mathsf{n_r}}
\def\td{\mathsf{t_d}}
\def\xd{\mathsf{x_d}}
\def\tbin{\mathsf{t_{bin}}}
\def\te{\mathsf{t_e}}
\def\tr{\mathsf{t_r}}
\def\kd{\mathsf{k_d}}

\def\fbf{\mathbf{f}}
\def\hbf{\mathbf{h}}
\def\gbf{\mathbf{g}}
\def\ubf{\mathbf{u}}
\def\vbf{\mathbf{v}}
\def\xbf{\mathbf{x}}
\def\ybf{\mathbf{y}}
\def\zbf{\mathbf{z}}
\def\Pbf{\mathbf{P}}
\def\lambf{\boldsymbol{\lambda}}

\def\A{\mathrm A}

\def\XA{X_{\mathrm A}}
\def\XD{X_{\mathrm D}}
\def\LF{\mathrm{LF}}
\def\HF{\mathrm{HF}}
\def\FIA{\text{FI}_{\mathrm A}}
\def\FID{\text{FI}_{\mathrm D}}

\def\lambdab{\lambda_{\mathrm b}}
\def\lambdahatb{\widehat{\lambda}_{\mathrm b}}
\def\lambdas{\lambda_{\mathrm s}}

\def\Mod{\; \mathsf{mod} \;}
\def\grad{\nabla}

\def\lam_b{\lambda_b}

\def\Ti{T_{i}}
\def\Tii{T_{i+1}}

\def\Bhat{\widehat{B}}
\def\Bmin{B_{\rm min}}
\def\Bml{\widehat{B}^{\rm ML}}
\def\Lhat{\widehat{\Lambda}}
\def\Lml{\widehat{\Lambda}^{\rm ML}}
\def\Shat{\widehat{S}}
\def\Smin{S_{\rm min}}

\newcommand{\Frac}[2]{{{#1}/{#2}}}  

\DeclareMathOperator*{\argmax}{arg\,max}

\newtheorem{proposition}{Proposition}

\newcommand{\diag}[1]{\text{diag} ( #1 )}

\begin{document}
\title{Dead Time Compensation for High-Flux Ranging}
\author{Joshua~Rapp, Yanting~Ma, Robin~M.~A.~Dawson, and~Vivek~K~Goyal
\thanks{J.~Rapp and Y.~Ma contributed equally to this work.}
\thanks{J.~Rapp, Y.~Ma, and V.~K~Goyal are with the Department
of Electrical and Computer Engineering, Boston University, Boston,
MA, 02215 USA e-mail: yantingm@bu.edu.}
\thanks{J.~Rapp and R.~M.~A.~Dawson are with The Charles Stark Draper Laboratory, Inc., Cambridge, MA 02139, USA.}
\thanks{This work was supported in part by a Draper Fellowship under Navy contract N00030-16-C-0014, the U.S. National Science Foundation under Grant No.\ 1422034 and Grant No.\ 1815896, and the DARPA REVEAL program under Contract No.\ HR0011-16-C-0030.}}

\maketitle

\begin{abstract}
Dead time effects have been considered a major limitation for fast data acquisition in various time-correlated single photon counting applications, since a commonly adopted approach for dead time mitigation is to operate in the low-flux regime where dead time effects can be ignored. 
Through the application of lidar ranging, this work explores
the empirical distribution of detection times in the presence of dead time and demonstrates that an accurate statistical model can result in reduced ranging error with shorter data acquisition time when operating in the high-flux regime.
Specifically, we show that the empirical distribution of detection times converges to the stationary distribution of a Markov chain.
Depth estimation can then be performed by passing the empirical distribution through a filter matched to the stationary distribution. 
Moreover, based on the Markov chain model, we formulate the recovery of arrival distribution from detection distribution as a nonlinear inverse problem and solve it via provably convergent mathematical optimization.  
By comparing per-detection Fisher information for depth estimation from high- and low-flux detection time distributions, we provide an analytical basis for possible improvement of ranging performance resulting from the presence of dead time.
Finally, we demonstrate the effectiveness of our formulation and algorithm via simulations of lidar ranging. 
\end{abstract}


\IEEEpeerreviewmaketitle


\section{Introduction}
\label{sec:intro}

Time-correlated single photon counting (TCSPC) is a powerful technique for measuring the fast, time-dependent responses of actively illuminated systems. 
Commonly used for fluorescence lifetime imaging (FLIM)~\cite{Becker2005}, TCSPC has also been applied to optical quantum information applications~\cite{Hadfield2009}, light detection and ranging (lidar)~\cite{Massa1997}, and non-line-of-sight (NLOS) imaging~\cite{Buttafava2015,OToole2018}, among others~\cite{Altmann2018}.
TCSPC is particularly useful for lidar because the single-photon
sensitivity allows for lower-intensity signal returns, either from weaker illuminations or from distant, oblique-angled, or otherwise non-cooperative targets~\cite{Pellegrini2000}.
As a result, TCSPC lidar is one of the competing technologies currently being developed for commercial autonomous vehicles.

One of the main hardware limitations of photon counting is that the instrumentation has a dead time, a period of insensitivity after a photon detection during which arriving photons cannot be registered.
Typical TCSPC applications use a laser repeatedly pulsed with period $\tr$ and build up a histogram of detection times relative to the most recent illumination.
Without compensation for the effects of dead time, detection time histograms appear as distorted versions of the incident light intensity waveform, leading to erroneous depth estimates.
The usual approach to dealing with dead time is to limit the acquisition flux so that photons are detected in at most $5\%$ of illumination periods; this reduces the probability of photons arriving during the dead time, thus limiting the number of arrivals that are ``lost.''
However, attenuating the flux incident on a detector is an inefficient use of the light reflected from a scene and slows down the acquisition of sufficient photons for accurate ranging.
Allowing for higher incident flux and compensating for the resulting distortions due to dead time would enable faster acquisition without loss of accuracy and has been the subject of several recent works~\cite{Heide2018,Pediredla2018}, although their models of dead time 
assume a single detection can be recorded per illumination period, an assumption which does not necessarily hold for modern timing electronics.
Our aim is to accurately model the effects of dead time on the photon detection process so that photons can be detected at a much higher rate and distortions introduced due to dead time can be predicted and corrected.
Eventually, this approach should lead to the possibility of higher laser powers, shorter acquisition time, and more accurate depth estimation.

\subsection{Dead Time Characterization}
The source of the dead time and resulting behavior of the system may vary greatly depending on the implementation. 
Our work studies dead time correction for modern TCSPC systems with asynchronous electronics (such as the HydraHarp 400~\cite{Wahl2008} or TimeHarp 260~\cite{Wahl2013} from PicoQuant) and a nonparalyzable detector (e.g., PDM-series~\cite{Giudice2007} or Fast-gated SPADs~\cite{Buttafava2014} from Micro Photon Devices, operated in free-running mode). 
In the following, we formally define paralyzability, the sources of the dead time, and synchronization so as to clarify the dead time model we consider.

\textbf{Paralyzability.}
The dead times of event-counting detectors have been studied since at least the 1940s~\cite{Gnedenko1941,Kosten1943,Levert1943,Kurbatov1945}, with Feller first classifying detectors in terms of their paralyzability~\cite{Feller2015}.
\emph{Nonparalyzable (Type I)} detectors are dead for a fixed time $\td$ after a detection, regardless of whether additional photons arrive during the dead time.
On the other hand, when photons arrive during the dead time of a \emph{paralyzable (Type II)} detector, the dead time restarts and extends for at least another $\td$.
We consider only nonparalyzable detectors in our work.

\textbf{Source.}
TCSPC systems suffer dead times from various components, each with different behavior.
The timing electronics in \emph{classical} TCSPC systems (in both nonreversed and reversed start-stop modes) are only able to record a single detection time per illumination period~\cite{Becker2005}.
\emph{Modern} TCSPC electronics 
allow for multiple detections per period, but the duration of the signal recording process still forces the timing electronics to be insensitive to additional photon arrivals~\cite{Arlt2013}.
In addition, the detectors themselves suffer from dead times.
For instance, the commonly used Single Photon Avalanche Diode (SPAD) detectors are reverse-biased photodiodes that are single-photon sensitive because they are operated above the breakdown voltage.
Incident photons cause an avalanche of carriers that is directly detectable as a precise digital signal, but which must be quenched in order to 
The time during which the SPAD is ``held-off'' is further extended to prevent afterpulses, which are avalanches caused by the release of trapped carriers from previous avalanches~\cite{Cova2013}. 
The total hold-off time
for active quenching circuits is thus a trade-off between shorter dead times versus lower afterpulsing 
probabilities; this parameter is often fixed in the detection circuitry design but may be left variable in some devices for tuning at the user's discretion (e.g.,~\cite{Buttafava2014}).
In this work, we assume that the hold-off 
time is sufficiently long such that any afterpulses can be considered indistinguishable from ambient detections and dark counts.\footnote{If the quenching time is too short, afterpulses can no longer be considered an independent Poisson process, as their occurrence is correlated with the previous detection time~\cite{Humer2015}.}

\textbf{Synchronization.}
The different sources of dead time further suggest two modes of TCSPC system operation.
We call systems \emph{synchronous} if they ensure that the end of a dead time is synchronized with the start of an illumination period.
Synchronous operation is often built into the hardware, such as in classical TCSPC systems or in the gated mode of fast-gated SPADs~\cite{Buttafava2014,Pediredla2018}.
In reversed start-stop mode, classical timing electronics may become active in the middle of an illumination period, but that recovery time is consistent, as the dead time is synchronized to the stop signal (either a delayed version of the current pulse or the next pulse).
Modern TCSPC electronics enable \emph{asynchronous} operation, in which there is no enforced synchronization between the dead time and the timing electronics.
If a photon is detected towards the end of a cycle and the dead time continues after the next laser excitation, there is no mechanism preventing the detector from becoming active in the middle of that cycle.
In other words, the end of the dead time is no longer dependent on the cycle synchronization, but on the arrival time of the most recently detected photon.
The synchronous and asynchronous architectures correspond to the ``clock-driven'' and ``event-driven'' SPAD recharge mechanisms, which were explored 
in~\cite{Antolovic2016,Antolovic2018}. 
While most existing work on the effect of dead time assumes synchronous operation, we consider only asynchronous systems in this work.

\subsection{Dead Time Compensation Approaches}
Yu and Fessler outline a number of general strategies for handling the effects of dead time~\cite{Yu2000}, with the simplest approach being to simply ignore the dead time.
Most commonly, the total photon flux at the detector is changed such that the dead time effects are actually negligible and can be ignored.
Since the effect of dead time is that photon arrivals within $\td$ of a detection are missed, a straightforward approach is to reduce the total photon flux, either by lowering the laser power and ambient light if possible, or by attenuating with a filter at the detector.
The suggestion of O'Connor and Phillips is to keep the fraction of excitations causing a detection to be at most 5\% 
to avoid dead time effects~\cite{OConnor1984}, a recommendation that electronics manufacturers have adopted.
Reducing the flux inevitably leads to longer time needed to acquire the same number of photons.
As a result, several recent works have focused on reducing the number of photons per pixel needed for accurate depth imaging by incorporating parameterized probability models of detection times and priors on the spatial structure of natural scenes~\cite{Kirmani2014,Shin2015,Altmann2016,Rapp2017}. 
Other approaches have tried to ignore dead time by changing the hardware setup, such as using multiple detectors so that there is more likely to be a detector not in the reset state when a photon arrives~\cite{Becker2005}.

Rather than attenuate the flux at the detector to avoid dead time effects, another strategy is to correct the distortions in the high-flux data after acquisition.
Most algorithm-based attempts at dead time compensation consider synchronous systems due not only to the systems that have historically been available, but also for the convenient property that detection times are statistically independent of each other in different cycles~\cite{Shin2015}.
One of the first methods for dealing with dead time in synchronous systems is that of Coates~\cite{Coates1968}.
Coates's basic algorithm was designed for lifetime measurement, with later work adapting the algorithm to include background subtraction~\cite{Davis1972}.
The basic principle of Coates's algorithm is that for any bin $i$ in a histogram, the detections in the preceding bins spanning $\td$ represent excitation cycles when no photon could have been detected in bin $i$ because the detector was dead.
The number of cycles in which the detector thus must have been dead is used to adjust computation of the photon arrival probability in each bin.
Recent work has rederived Coates's expression, which is the ML estimator for the number of photon arrivals in each bin of a histogram in a synchronous system, in order to include priors for maximum a posteriori (MAP) estimation~\cite{Pediredla2018}.
A few models \cite{Xu2016,Verma2017} consider histogram corrections for a hybrid of synchronous and asynchronous systems, which do allow for multiple statistically dependent detections per illumination cycle but without the dependency carrying over into different cycles.
A handful of papers address special cases of dead time effects in asynchronous systems: Antolovic et al. consider detection rate estimation for homogeneous arrival processes in~\cite{Antolovic2016,Antolovic2018}, whereas Cominelli et al. explore the special case when $\td$ equals an integer multiple of $\tr$ and no correction is needed~\cite{Cominelli2017}.
However, these approaches are insufficient to address the typical lidar acquisition mode with inhomogenous arrivals and in which $\td$ and $\tr$ cannot necessarily be arbitrarily adjusted.
The only work the authors are aware of that addresses asynchronous systems generally is that of Isbaner et al.~\cite{Isbaner.etal2016}, which effectively treats the detection process as a time-dependent attenuation of the arrival process intensity.
Although they model both the electronics and detector dead times $\te$ and $\td$, respectively, they note that such a system simplifies to having only one source of dead time when $\te<\td$.
They use an iterative procedure to estimate the attenuation, which is then used to correct the detection histogram.

The last strategy for dealing with dead time is to use the data as acquired but to incorporate dead time into the detection model.
In this vein, Heide et al. adjust their parameter estimation procedure to include dead time effects~\cite{Heide2018}.
However, the synchronous system assumption they use is technically only valid for their asynchronous timing electronics (PicoQuant PicoHarp 300,~\cite{Wahl2007}) if zero ambient light is present, which
guarantees that the detector will be reset for the next signal pulse.

\subsection{Other Related Work}
In addition to missing photons that arrive during a dead time, events can fail to be registered in the recording of point processes through other forms of ``counting loss,'' which depends on the measurement system design~\cite{Becker2005}.
A well-documented form of counting loss is ``pile-up,'' referring to the problem 
of the rising edge of a pulse overlapping with the tail of a previous pulse, such that the later pulse is not registered by a discriminator.
Each piled-up pulse prolongs the duration in which new events cannot be detected, making the discriminator a Type II detector.
Pile-up is present in some TCSPC system designs, such as those using passively-recharged SPADs~\cite{Antolovic2018} or hybrid photodetectors with negligible dead time~\cite{Patting2018}.
Several approaches correcting for pile-up have recently been proposed for nuclear spectrometry, in which  both the pulse times and energies are of interest~\cite{Trigano2005,Sepulcre2013,Trigano2017}.
Confusingly, the term ``pile-up'' is also sometimes used to describe the effect of dead time in synchronous TCSPC systems (e.g.,~\cite{Pediredla2018,Heide2018}) because the effect of earlier detections preventing later detections is similar.
Another form of counting loss, named Type III in~\cite{Yu2000}, may also occur in some systems (e.g.,~\cite{Pomme2015}) when two pulses occur close together and neither one gets recorded.

In lidar applications, the dead time-affected acquisition results in closer apparent distances, which has a similar effect to the intensity-dependent change in perceived depth known as ``range walk error''
~\cite{Palojarvi2005,Kurtti2009,Cho2014}.
Range walk is the result of using a discriminator to trigger in the leading edge of a signal pulse; a stronger signal with a steeper rising edge will be detected earlier than a weaker signal with a smaller slope.
Due to the similarity with high-flux ranging, approaches correcting for range walk error could be adapted to compensate for dead time.
Several optics-based methods aimed at range walk error correction attempt to experimentally measure and then correct for the bias in depth estimation. 
He et al.\ first calibrate the amount of range walk incurred as a function of the detection rate~\cite{He2013}.
Then the conventional depth estimation procedure is performed with the dead time-distorted data, and the range walk bias is subtracted off to correct the depth estimate.
Ye et al.\ use a similar method, except they first split the incident light with a 90:10 beamsplitter to two SPADs, using the lower-flux channel for simultaneous bias estimation to subtract off from the lower-variance high-flux estimate~\cite{Ye2018}.

\subsection{Main Contributions}
\subsubsection{Markov Chain Detection Time Model} 
We rigorously construct a Markov chain model to characterize the empirical distribution of detection times in asynchronous TCSPC systems.
Analyzing the stationary distribution of the Markov chain directly leads to a simple log-matched filter estimator for depth estimation.
\subsubsection{Arrival Intensity Reconstruction}
We derive a nonlinear inverse formulation for arrival intensity estimation from the detection distribution, where the formulation is based on the stationary condition of the Markov chain and the nonlinear inverse problem is solved by a provably convergent optimization algorithm; the estimated arrival intensity can then be used for depth estimation and other tasks.
\subsubsection{Accurate High-Flux Ranging}
Using our Markov chain-based methods, we show that depth estimation from high-flux detection data can achieve lower error than using low-flux data for the same acquisition time or can alternatively achieve the same error from much faster acquisitions. 
Furthermore, our methods outperform the method of~\cite{Isbaner.etal2016} applied to high-flux detection data.
\subsubsection{Demonstration of Dead Time Benefits}
By comparing the Fisher information per detected photon for depth estimation from the high- and low-flux detection time distributions, we show that when the background rate is low and the signal rate is sufficiently high, the presence of dead time may lead to improvement in ranging accuracy for a fixed number of detections. 
Moreover, when the dead time $\td$ is only slightly smaller than the illumination period $\tr$, such improvement can extend to higher background rate scenarios, since the dead time acts as a signal-triggered gate in this case.

\section{Empirical Distribution of Detection Times}
\label{sec:formulation}

The challenge of studying the detection time distribution for the asynchronous dead time model is that the detection times are statistically dependent. 
In this section, we show that the dependency is Markovian and provide the explicit transition probability density function (PDF). From the transition PDF, we can analyze the stationary condition and obtain the stationary distribution, from which our high-flux ranging algorithms are derived.

\subsection{Photon Arrival Process}
\label{subsubsec:arriv_process}

It is well known that photon arrival times at a detector are described by a Poisson process~\cite{Snyder.Miller2012}. 
For TCSPC, the repeated illumination with period $\tr$ makes the arrival process an inhomogeneous Poisson process with periodic intensity function $\lambda(t)$.
In general, $\lambda(t)$ is composed of two parts:
\begin{equation}
\lambda(t) = \lambdas(t) + \lambdab(t),
\label{eq:def_lam}
\end{equation}
where $\lambdas(t)$ is the time-varying intensity of a \emph{signal} process and $\lambdab(t)$ is the intensity due to \emph{background} (ambient light), which is assumed to be constant $\lambdab$ in this work.
For the application of ranging, $\lambdas(t)$ is often described parametrically in one period as the scaled and time-shifted illumination pulse $\lambdas(t) = \alpha \beta s(t-2z/c)$, 
where 
$\alpha$ is the target reflectivity, $z$ is the target  depth,
$\beta$ is a gain factor corresponding to the illumination power, $s(t)=\Frac{\exp(-\Frac{t^2}{2\sigma^2})}{\int_0^{\tr} \exp(-\Frac{\tau^2}{2\sigma^2}) \d \tau}$ is the Gaussian pulse shape approximation with half pulse width $\sigma$,
and $c$ is the speed of light.
Within one period, the signal photon arrival rate $S$ and background photon arrival rate $B$ are defined as $S := \int_0^{\tr} \lambdas(\tau) \, \d\tau$ and $B :=  \lambdab {\tr}$, respectively.
The total flux is given by $\Lambda:=S+B$, and the signal to background ratio is defined as $\text{SBR}:=S/B$.

\subsection{Markov Chain Model for Detection Times}
\label{subsec:markov}
If there were no dead time effects, the detection process would be equivalent to the arrival process, which is Poisson with intensity $\lambda(t)$; conditioned on the total number of detections, the absolute detection times would be order statistics of i.i.d. random variables with common probability density function $\propto\lambda(t)$~\cite[Section 2.3.3]{Snyder.Miller2012}.%
\footnote{Throughout the paper, \emph{absolute detection time} refers to the time when periodicity is not taken into consideration (i.e., with $\nr$ illumination cycles, we have $0<t_1\leq t_2,\ldots,\leq t_{N(t)} \leq \nr\tr$), whereas \emph{detection time} refers to the time of detection relative to the most recent illumination pulse, which is the absolute detection time modulo $\tr$.
Therefore, the limiting empirical distribution of detection times would be $\propto\lambda(t)$.
}
However, in the presence of dead time effects, the detection process is no longer Poisson, since the detection intensity, denoted by $\mu(t)$, is now a random process depending on the history of the detection process; such a detection process is referred to as a self-exciting process~\cite{Snyder.Miller2012}. 
Specifically, let $\{N(t):t\geq 0\}$ denote the detection process with (random) intensity $\mu(t)$, where $\{N(t):t\geq 0\}$ is characterized by the number of detections $N(t)$ at time $t$ and a sequence of absolute detection times $T_1,\ldots,T_{N(t)}$.
The conditional PDF of $T_{i+1}$ given $T_{1},\ldots,T_i$ is \cite[(6.15) (6.16)]{Snyder.Miller2012}
\begin{equation}
f_{T_{i+1}|T_1,\ldots,T_i}(t|t_1,\ldots,t_i) = \mu(t)\exp\!\left(-\int_{t_i}^{t}\mu(\tau) \, \d\tau\right).
\label{eq:con_density_T0}
\end{equation}
For a general self-exciting process, $\mu(t)$ can depend on the entire history of the process $\{N(\tau):0\leq \tau<t\}$. For the specific detection process considered in this work, we have
\begin{equation}
\mu(t)=\begin{cases}
\lambda(t), &\text{if } t> T_{N(t)} + \td; \\
0, &\text{if } T_{N(t)} < t \leq T_{N(t)}+\td,
\end{cases}
\label{eq:def_mu}
\end{equation}
where we introduce the notation $T_0 := -\infty$. We can see that $\mu(t)$ only depends on the latest detection time. Therefore, for the $\mu(t)$ defined in \eqref{eq:def_mu}, the RHS of \eqref{eq:con_density_T0} depends on $t_i$ but not on $t_1,\ldots,t_{i-1}$. That is, the absolute detection times form a Markov chain with transition PDF
\begin{equation}
   f_{T_{i+1}|T_i}(t|t_i) = \lambda(t)\exp\!\left(\!-\!\int_{t_i+\td}^{t}\!\!\!\lambda(\tau) \, \d\tau\!\right)\mathbb{I}\{t>t_i+\td\}, 
   \label{eq:con_density_T}
\end{equation}
where $\mathbb{I}$ is the indicator function.
An illustration of a realization of the detection process is shown in Fig.~\ref{fig:process_intensities}.

\begin{figure}
	\centering
	\includegraphics[width=\linewidth]{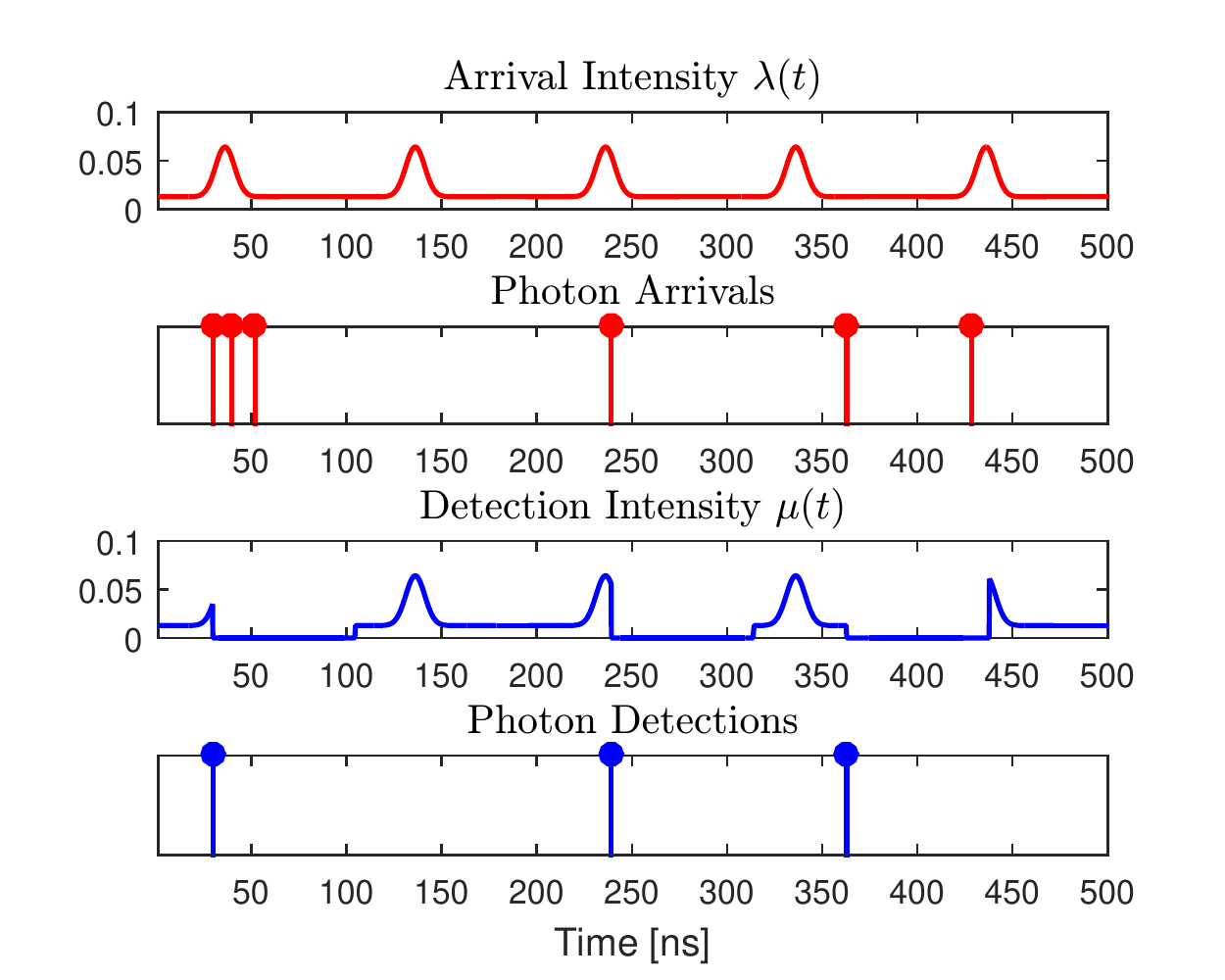}
	\caption{Illustration of the effect of dead times on the detection process for $\tr = 100$ ns, $\td = 75$ ns, $\sigma = 5$ ns, $S$ = 0.5, and $B$ = 1. 
	Photon arrival times are generated according to the arrival intensity $\lambda(t)$.
	The detection intensity $\mu(t)$ is equal to the arrival intensity $\lambda(t)$ except immediately following a photon detection when $\mu(t) = 0$, so the detection times are a subset of the arrival times and detection is not a Poisson process.}
	\label{fig:process_intensities}
\end{figure}

Define two sequences of random variables,
$\{K_i\}_{i\in\mathbb{N}}$ and $\{X_i\}_{i\in\mathbb{N}}$,
such that $K_i:= \lfloor T_i/\tr \rfloor$, where $\lfloor a\rfloor$ is the integer part of $a\in\mathbb{R}$, and $X_i:=T_i \Mod \tr$, hence $T_i = K_i\tr + X_i$.
That is, $K_i$ is the number of illumination periods before $T_i$ 
and $X_i$ is the location of absolute detection time $T_i$ within illumination period $K_i+1$, which is referred to as detection time in this paper.
Note that if there were no dead time effects, the empirical distribution of $X_i$'s would be identical to the arrival time PDF, given by
\begin{equation}
f_{\XA}(x)=\lambda(x)/\Lambda, \quad\text{for } x\in [0,\tr).
\label{eq:def_fXa}
\end{equation}
The following proposition provides statistical characterization of $\{X_i\}_{i\in\mathbb{N}}$ in the presence of dead time.

\begin{proposition}
Suppose that the photon arrival process is an inhomogeneous Poisson process with periodic intensity function $\lambda(t)$, whose period is $\tr$, and the detector has dead time $\td$. Define $\xd:=\td \Mod \tr$.
Let the random sequence $\{T_i\}_{i\in\mathbb{N}}$ denote absolute detection times and define detection times as $X_i:= T_i \Mod \tr$, for all $i\in\mathbb{N}$. Then the random sequence $\{X_i\}_{i\in\mathbb{N}}$ forms a Markov chain with state space $[0,\tr)$ and transition PDF
\begin{align}
&f_{X_{i+1}|X_i}(x_{i+1}|x_i) = \frac{\lambda(x_{i+1})}{1-\exp(-\Lambda)}\nonumber \\
&\qquad\exp\Bigg(-\displaystyle\int_{x_i + \xd}^{\Big\lceil \displaystyle\frac{x_i+\xd-x_{i+1}}{\tr} \Big\rceil \tr + x_{i+1}} \lambda(\tau) \, \d\tau \Bigg),
\label{eq:cond_pdf}
\end{align}
where $\lceil a \rceil := \lfloor a \rfloor + 1$ and $\Lambda:=\int_{0}^{\tr}\lambda(\tau) \, \d\tau$.
\label{prop_markov}
\end{proposition}
\begin{proof}
See Appendix \ref{app_proof_markov}.
\end{proof}
We can check that $\{X_i\}_{i\in\mathbb{N}}$ is $\psi$-irreducible, recurrent, and aperiodic, and hence it has a unique stationary PDF \cite[Proposition 10.4.2]{Meyn.Tweedie2012}.
Denoting the stationary PDF by $f_{\XD}$, then for all $x\in [0,\tr)$, $f_{\XD}$ satisfies 
\begin{equation}
f_{\XD}(x) = \int_{0}^{\tr} f_{\XD}(y)f_{X_{i+1}|X_i}(x|y) \, \d y.
\label{eq:stationary_cond}
\end{equation}
That is, $f_{\XD}$ is the eigenfunction corresponding to eigenvalue 1 of the linear operator $\mathcal{P}$ defined as
\begin{equation}
\mathcal{P}f(x):=\int_{0}^{\tr} f(y)f_{X_{i+1}|X_i}(x|y) \, \d y.\label{eq:trans_operator}
\end{equation}
For the special case where $\xd=0$, we show in the following that the arrival PDF $f_{\XA}$ defined in \eqref{eq:def_fXa} satisfies the stationary condition, meaning that dead time does not cause any distortion in detection time distribution; this result has also been noted in~\cite{Cominelli2017} with a different derivation. With $f_{\XD}(x)=\lambda(x)/\Lambda$ and $\xd=0$, the RHS of \eqref{eq:stationary_cond} is
\begin{align*}
&\int_{0}^{\tr} \frac{\lambda(y)}{\Lambda}f_{X_{i+1}|X_{i}}(x|y) \d y\\
&= \int_{0}^{x} \frac{\lambda(y)}{\Lambda}\frac{\lambda(x)}{1-\exp(-\Lambda)}\exp\left(-\int_{y}^{x}\lambda(\tau)\d \tau\right) \\
&\qquad\qquad+ \int_{x}^{\tr} \frac{\lambda(y)}{\Lambda}\frac{\lambda(x)}{1-\exp(-\Lambda)}\exp\left(-\int_{y}^{\tr+x}\lambda(\tau)\d \tau\right) \\
&=\frac{\lambda(x)}{\Lambda}\Bigg[\frac{\int_{0}^{x}\lambda(y)\exp\left(-\int_{y}^{x}\lambda(\tau)\d \tau\right)\d y}{1-\exp(-\Lambda)}\\
&\qquad\qquad\qquad\qquad+\frac{\int_{x}^{\tr}\lambda(y)\exp\left(-\int_{y}^{\tr+x}\lambda(\tau)\d \tau\right)\d y }{1-\exp(-\Lambda)}\Bigg].
\end{align*}
Label the two terms in the square brackets as $A_1$ and $A_2$. Using the chain rule and the Leibniz rule for differentiation, for any constant $a$ that does not depend on $y$, we have that
\[\frac{\d}{\d y}\exp\left(-\int_{y}^{a}\lambda(\tau)\d \tau\right) = \lambda(y)\exp\left(-\int_{y}^{a}\lambda(\tau)\d \tau\right).\]
Letting $a=x$, we have 
\begin{equation*}
    A_1 =\frac{\exp\left(-\int_{y}^{x}\lambda(\tau)\d \tau\right)\Big\vert_0^x}{1-\exp(-\Lambda)}=\frac{1-\exp\left(-\int_{0}^{x}\lambda(\tau)\d \tau\right)}{1-\exp(-\Lambda)}.
\end{equation*}
Similarly, let $a=\tr+x$, then
\begin{equation*}
     A_2 =\frac{\exp\left(-\int_{0}^{x}\lambda(\tau)\d \tau\right)-\exp(-\Lambda)}{1-\exp(-\Lambda)}.   
\end{equation*}
It follows that $A_1+A_2=1$, and so  $\Frac{\lambda(x)}{\Lambda}$ is the stationary distribution of the Markov chain when $\xd=0$.

To numerically demonstrate the correctness of \eqref{eq:cond_pdf} for general $\xd$, 
we partition the state space $[0,\tr)$ into $\nb$ equally spaced time bins with bin centers $\{b_n\}_{n=1}^{\nb}$ and approximate the linear operator $\mathcal{P}$ defined in \eqref{eq:trans_operator} with an $\nb\times\nb$ matrix $\Pbf$, where $P_{m,n}:=f_{X_{i+1}|X_i}(b_n|b_m)$ with $f_{X_{i+1}|X_i}$ defined in \eqref{eq:cond_pdf}. 
The matrix $\Pbf$ is then normalized to have row sum equal to 1 so that it becomes a probability transition matrix $\widetilde{\Pbf}$. 
A discrete approximation of $f_{\XD}$, denoted by a length-$\nb$ row vector $\fbf_{\XD}$, is then obtained as the leading left eigenvector of $\widetilde{\Pbf}$, since $\fbf_{\XD}$ should satisfy the Markov chain stationary condition $\fbf_{\XD}= \fbf_{\XD} \widetilde{\Pbf}$. Moreover, if the second largest (in terms of magnitude) eigenvalue of $\widetilde{\Pbf}$ is strictly less than one, in other words, $\widetilde{\Pbf}$ admits a spectral gap, then the corresponding Markov chain converges to its stationary distribution geometrically fast. 
We have verified that in all parameter settings considered in this paper, $\widetilde{\Pbf}$ admits a spectral gap, thus confirming the convergence of the chain.
Finally, $\fbf_{\XD}$ is compared with the histogram of a set of simulated detection times, where we expect a close match between the simulated histogram and $\fbf_{\XD}$.
Detection times are simulated by first generating arrival times according to~\eqref{eq:def_lam}.
Then starting with detection of the first arrival time generated, subsequent arrivals are culled from the sequence if they are within $\td$ of the previous absolute detection time, as in~\cite{Hernandez2017}.
Note that unlike in~\cite{Hernandez2017}, both background photons and dark counts are considered to trigger dead times in the same manner as signal detections.

\begin{figure*}
\centering
    \begin{subfigure}{0.01\linewidth}
        \begin{sideways}
            	\text{\large \hspace{2mm}  $\mathbf{\mathsf{t_r} = 100}$ \hspace{15mm} $\mathbf{\mathsf{t_r}=80}$}
        \end{sideways}
    \end{subfigure} 
    ~
\begin{subfigure}{0.97\textwidth}
\includegraphics[width=\textwidth]{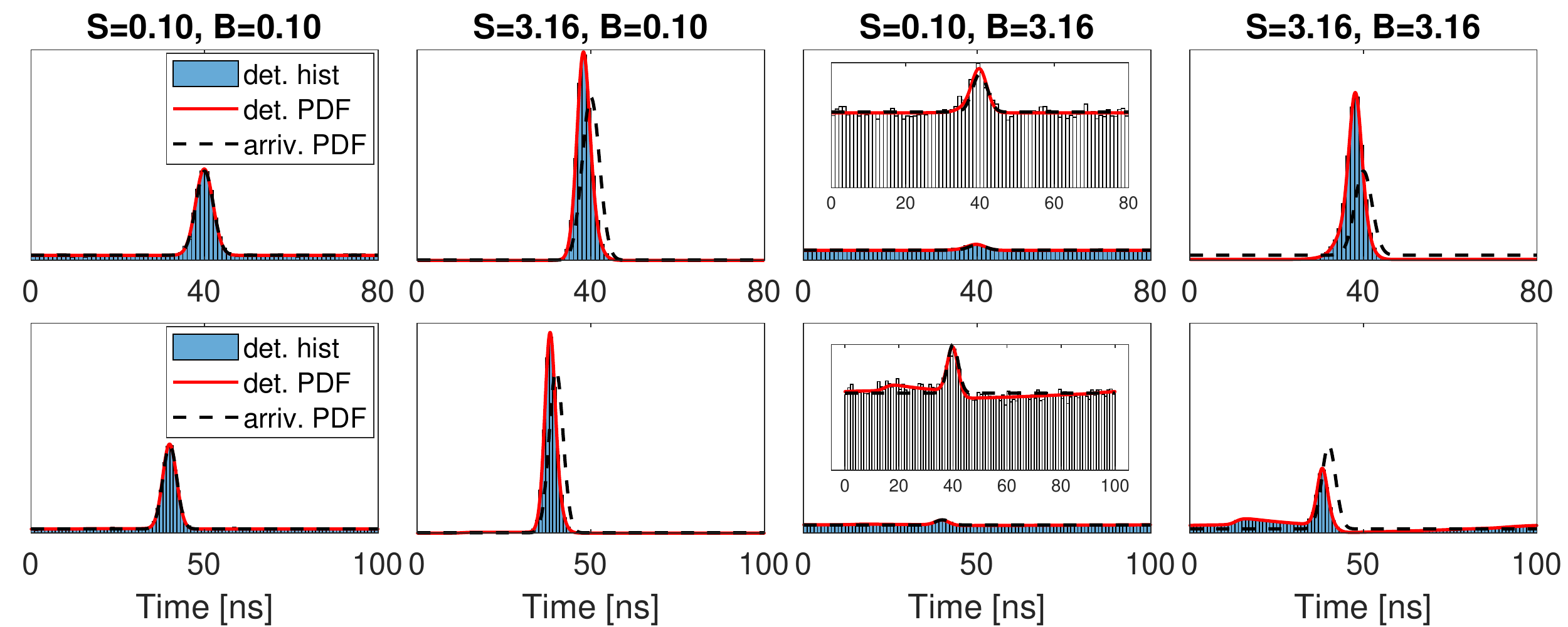}
\end{subfigure}
\caption{Comparisons between histograms of simulated detection times, predicted detection time PDFs, and arrival time PDFs illustrate how dead time affects the detection process. 
In addition to a shift in the mode toward earlier detection times, the dead time may also cause a ripple in the detection PDF relative to the arrival PDF.
Plots are shown for $\sigma=2$ ns, $\tbin = 50$ ps, $\nr=50000$, and $\td=75$ ns. 
The vertical axis scale is constant for each row.
An inset with a different vertical scale is included for each plot in the third column to emphasize the ripple that is not easily seen in the original scale.}

\label{fig:density_hist}
\vspace*{-3mm}
\end{figure*}

Comparisons between a histogram of detection times collected from simulation and the corresponding $\fbf_{\XD}$ are shown in Fig.~\ref{fig:density_hist}.
In each simulation, the number of illuminations is $\nr = 50000$ and the half pulse width is $\sigma = $ 2 ns.
The close matches between predicted detection PDFs and the simulated histogram results validate the effectiveness of the Markov chain model in deriving the limiting distribution.
The figure further illustrates the effect that dead time has on TCSPC\@.
The first column of Fig.~\ref{fig:density_hist} shows results with $S=B=0.1$, so the total flux $\Lambda$ is low enough that
few photons arrive during the detector dead time, and the arrival and detection densities are almost identical. 
If just the signal flux is increased, e.g., by increasing the illumination laser power (second column), the photon detection density narrows and shifts slightly toward earlier detection times (similar to the phenomenon of range walk error), due to early arrivals from the pulse blocking later photons from being detected. 
When the background flux increases, the distortions in the density due to dead time become more apparent.
However, these distortions also depend on the particular values of $\tr$ and $\td$.
When $\tr$ is slightly larger than $\td$ (such as for $\tr=80$ in the top row of Fig.~\ref{fig:density_hist}), 
the dead time triggered after a signal detection will reset just before signal photons from the next pulse arrive at the detector.
The dead time thus behaves as a signal-triggered gate, blocking detection of many background photons while allowing detection of additional signal.
On the other hand, increasing $\tr$ by just 20 ns (bottom row) causes a significant ripple in detection PDF a duration $\td$ after the main signal peak (modulo $\tr$).
The dead time is again often triggered by signal photons when $S$ is large, but the reset of the detector in the next cycle allows incident background photons to be detected, amplifying the apparent background intensity at that part of the cycle.
Note that this pre-pulse ripple could easily be mistaken for optical system inter-reflections or poor electronics thresholding if detector dead time were not taken into account.

\subsection{Comparison of Fisher Information}
\label{subsec:fisher}
High-flux acquisition enables detection of more photons than low-flux acquisition for a fixed number of illuminations. 
Although the detection time distribution is distorted in the sense that it is different from the arrival time distribution, our Markov chain model allows us to accurately predict the distortion. 
Therefore, it is expected that for a fixed number of illuminations, high-flux acquisition with our probabilistic model for detection times can improve ranging performance over the 5\% low-flux acquisition rule.
Another interesting aspect is to compare estimates from low-flux and high-flux acquisitions for a fixed number of detections. 
By comparing the arrival PDF, which is equivalent to the low-flux detection PDF, and the high-flux detection PDF in Fig.~\ref{fig:density_hist}, we notice that 
dead time 
results in a ``narrowing" of the pulse, especially for large $S$. 
Thus, we speculate that the distortion may in fact be favorable for depth estimation in some cases.

\begin{figure*}[t]
\centering
    \begin{subfigure}{0.01\linewidth}
        \begin{sideways}
            	\text{\normalsize B}
        \end{sideways}
    \end{subfigure} 
    ~
\begin{subfigure}{0.235\textwidth}
\includegraphics[width=\textwidth]{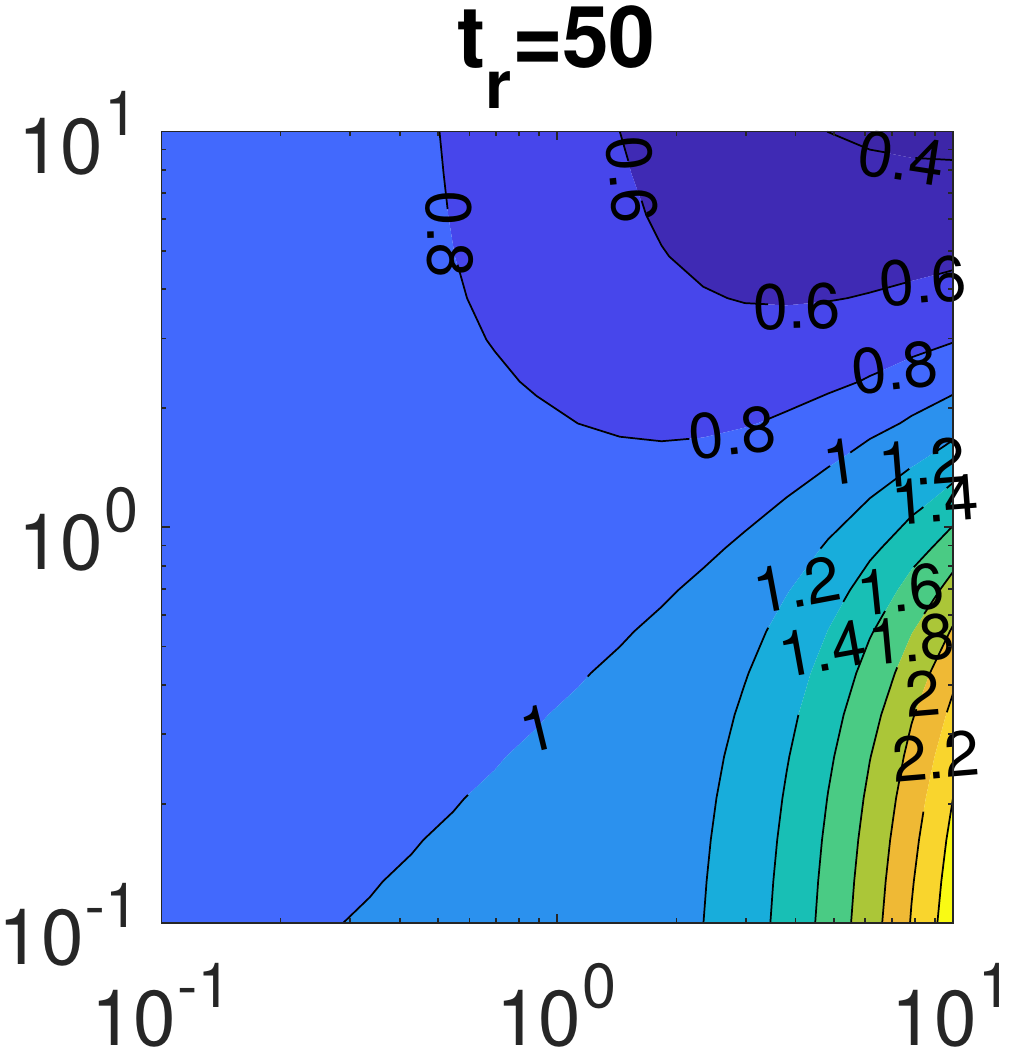}
\caption*{\normalsize S}
\end{subfigure}
\begin{subfigure}{0.235\textwidth}
\includegraphics[width=\textwidth]{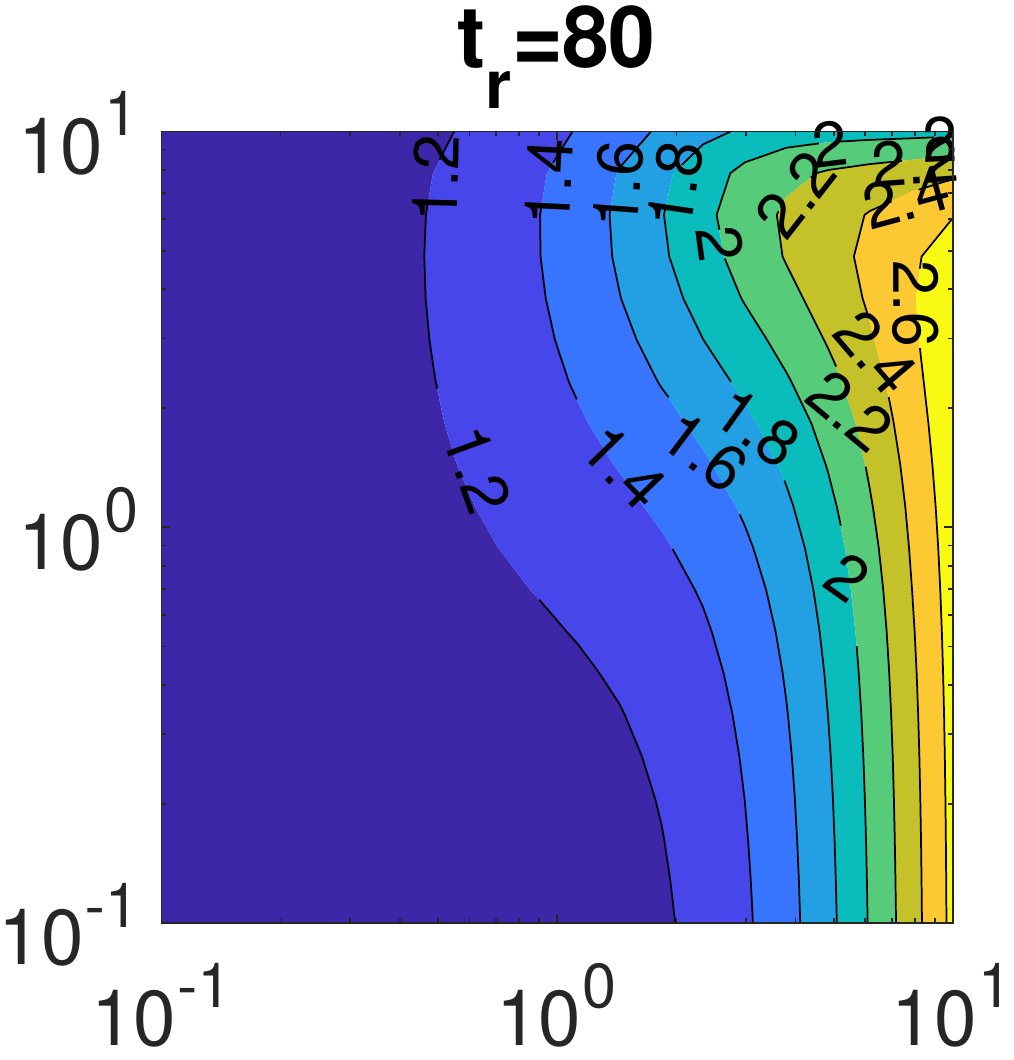}
\caption*{\normalsize S}
\end{subfigure}
\begin{subfigure}{0.235\textwidth}
\includegraphics[width=\textwidth]{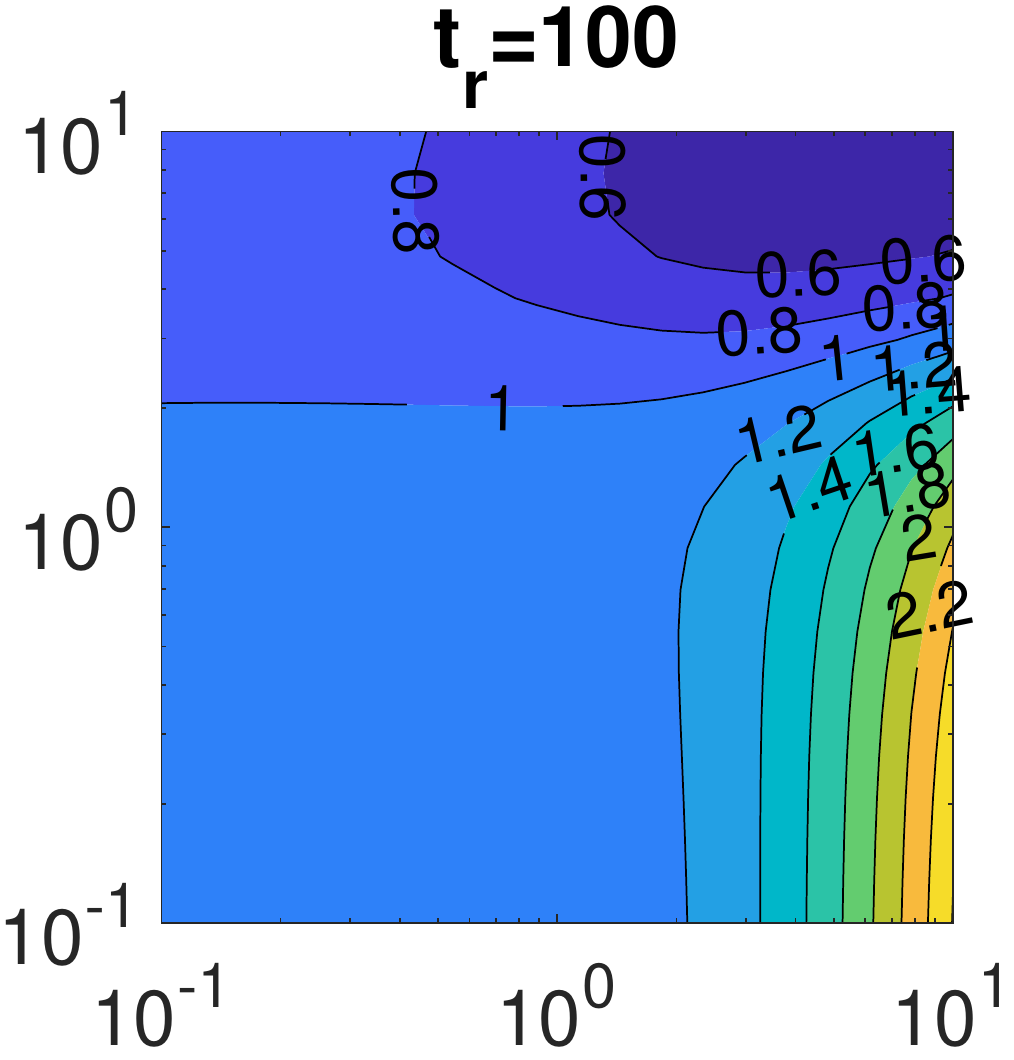}
\caption*{\normalsize S}
\end{subfigure}
\begin{subfigure}{0.235\textwidth}
\includegraphics[width=\textwidth]{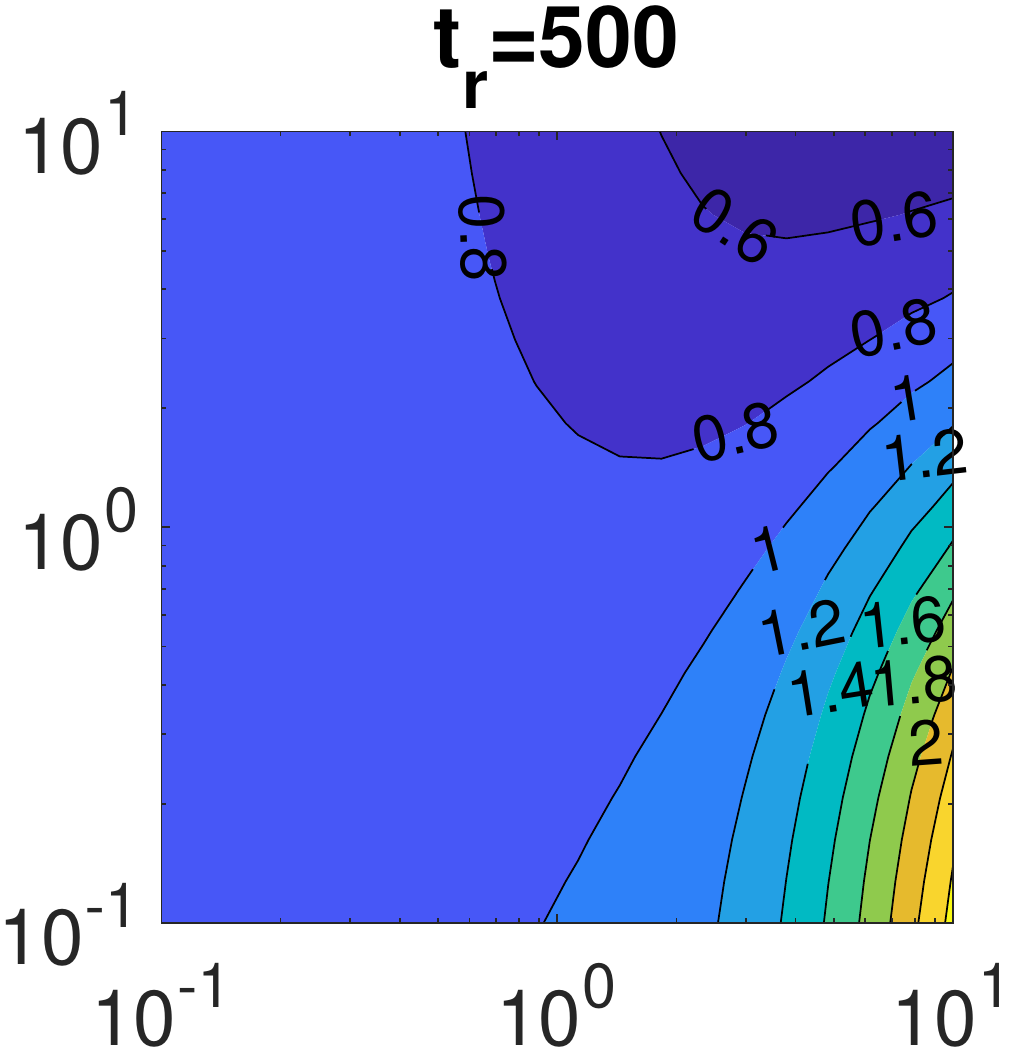}
\caption*{\normalsize S}
\end{subfigure}
\caption{The Fisher information ratio $\FID/\FIA$ indicates the performance improvement that may be gained for the same number of detections when high-flux data is used instead of low-flux data. The plots show for various signal rate $S$ and background rate $B$ and for $\sigma=0.2$ ns, $\tbin=10$ ps, and $\td=75$ ns that when SBR is sufficiently high and $B$ is not too large, the 
effect of dead time is beneficial for range estimation.}
\label{fig:FisherInfo}
\end{figure*}

To verify this somewhat counter-intuitive speculation, we compare the Fisher information per detection for estimating the depth $z$ from the low-flux PDF $f_{\XA}$~\eqref{eq:def_fXa} and high-flux PDF $f_{\XD}$ (stationary distribution of the Markov chain defined in Proposition~\ref{prop_markov}),\footnote{The reciprocal of the Fisher information is a lower bound for the mean squared error (MSE) of any unbiased estimator; this is also known as the Cram{\'e}r-Rao inequality~\cite[Theorem 3.1]{Kay:93v1}. 
} which are denoted by $\FIA$ and $\FID$, respectively, and are computed as (derivation is provided in Appendix~\ref{app_fisher_info})
\begin{equation*}
\begin{split}
\FIA &= \int_{0}^{\tr} \left(\frac{\partial }{\partial z}f_{\XA}(x;z)\right)^{\!2} \frac{1}{f_{\XA}(x;z)} \, \d x, \\
\FID &= \int_{0}^{\tr} \left(\frac{\partial }{\partial z}f_{\XD}(x;z)\right)^{\!2} \frac{1}{f_{\XD}(x;z)} \, \d x,
\end{split}
\end{equation*} 
where the derivative of $f_{\XD}$ is computed numerically. Note that while realizations of detection times are not i.i.d. samples of $f_{\XA}$ or $f_{\XD}$, most ranging algorithms only use the empirical distribution of detection times for depth estimation. Therefore, it is reasonable to consider Fisher information of the limiting empirical distributions $f_{\XA}$ and $f_{\XD}$ rather than that of the joint distributions.
Fig.~\ref{fig:FisherInfo} presents the Fisher information ratio $\FID/\FIA$ for $\td=75$ ns and with $\tr$ varying from 50 to 500 ns. 
By \eqref{eq:cond_pdf}, we notice that only $\xd:=\td \mod \tr$ affects the detection time distribution. 
Hence, for the case where $\tr=50$ ns, the effective dead time in terms of detection time distribution is 25 ns.
In the regions where the ratio is greater than one, $f_{\XD}$ is more informative about the depth $z$ than $f_{\XA}$ (i.e., the dead time effect is beneficial) in the sense that the per-detection Fisher information is higher.
We notice that such a region usually appears when $B$ is not too large and SBR is sufficiently high. 
A potential reduction in depth error variance was likewise noted by Heide et al.~\cite{Heide2018}, but that analysis assumed zero background, which is a na\"ive assumption for most applications and which our analysis shows is not a necessary condition for dead time to be beneficial.
When $\tr$ is slightly larger than $\td$ (as for $\tr=80$ in Fig.~\ref{fig:FisherInfo}), 
the signal-triggered gating extends the region in which dead time is beneficial to larger $B$ compared to the cases where $\tr$ is much larger than $\td$.
Together with the plots in Fig.~\ref{fig:density_hist}, this suggests that 
the most photon-efficient benefit from dead time is achieved when $\tr$ is slightly larger than $\td$.\footnote{Note that more photons would be detected with shorter $\td$, but each detection would likely be less informative of the depth.}
This condition may be difficult to achieve in practice as the dead time is not tunable in many devices and adjustment of the illumination period is limited by the required maximum unambiguous range.

\section{Arrival Intensity Estimation Algorithm}
\label{sec:algo}

In this section, we first derive a maximum likelihood (ML) estimator for estimating the total flux $\Lambda$ from absolute detection times. Then we derive an algorithm for estimating the arrival intensity $\lambda(x)$ for $x\in [0,\tr)$ from a histogram of detection times assuming that $\Lambda$ is known; one may implement our algorithm with a calibrated $\Lambda$ when available or with a $\Lambda$ estimated by, for example, our ML estimator.

\subsection{Maximum Likelihood Estimator for $\Lambda$}
\label{subsec:Lam_est}
In~\cite{Isbaner.etal2016}, Isbaner et al. note the necessity of estimating $\Lambda$ in order to correctly reconstruct the histogram of photon arrival times.
Define the interdetection period $R_i$ as the number of completed periods after the detector reset at $\Ti+\td$ before another photon is detected at time $T_{i+1}$:
\begin{equation}
    R_i := \left \lfloor \frac{\Tii-(\Ti+\td)}{\tr} \right  \rfloor,
    \label{eq:K_def}
\end{equation}
where $T_i$'s are absolute detection times. 
Isbaner et al. claim that $P(R_i=r) \propto \exp(-r\Lambda)$ and use weighted least squares to fit an exponential function.
In the following proposition, we verify the claim using properties of Poisson processes and the Markov nature of detections with dead time. 
Moreover, we show that $R_i$'s are independent, and so an ML estimator for $\Lambda$ can be easily computed from a realization of $R_i$'s.

\begin{proposition}
\label{prop:K_iid}
The random variables $R_i$'s defined in \eqref{eq:K_def} are i.i.d.\ with the same probability distribution as $R$, where
\begin{equation}
    P(R=r)=\left(1-\exp(-\Lambda)\right)\exp(-r\Lambda),\hspace{2mm} r\in\{0\}\cup\mathbb{N}.\label{eq:K_pdf}
\end{equation}
\end{proposition}
\begin{proof}
See Appendix~\ref{app_proof_K_iid}.
\end{proof}
By Proposition~\ref{prop:K_iid}, given a realization of interdetection periods  $\{r_i\}_{1\leq i\leq n}$, the log-likelihood function is 
\begin{equation*}
 \mathcal{L}(\{r_i\}_{i=1}^n;\Lambda)=-\Lambda\sum_{i=1}^n r_i + n\ln\left(1-\exp(-\Lambda)\right). 
\end{equation*}
Setting the derivative of $\mathcal{L}(\{r_i\}_{i=1}^n;\Lambda)$ with respect to $\Lambda$ to zero, we obtain the ML estimator for $\Lambda$ as
\begin{equation}
    \Lml = -\ln \left (\frac{\sum_{i=1}^n r_i}{n+ \sum_{i=1}^n r_i} \right ).
    \label{eq:ML_Lam}
\end{equation}
Note that the distribution of $R$ can be understood as follows. 
The number of photon arrivals per period is Poisson with parameter $\Lambda$, so $p = 1-\exp(-\Lambda)$
is the probability of at least one photon arriving in a period. Then $R$ has a geometric distribution $P(R=r) = (1-p)^r p$, which matches \eqref{eq:K_pdf}. 

In what follows, we consider estimating the arrival intensity $\lambda$ from a detection time histogram assuming that $\Lambda$ is known. 

\subsection{Relationship between Arrival and Detection Distributions}
Plugging \eqref{eq:cond_pdf} into \eqref{eq:stationary_cond}, we have 
\begin{align}
f_{\XD}(x) &= \lambda(x)\Bigg[\int_0^{x-\xd} f_{\XD}(y)\frac{\exp\!\left(-\int_{y + \xd}^{x} \lambda(\tau) \, \d\tau \right)}{1-\exp(-\Lambda)} \, \d y\nonumber\\
&\quad + \int_{x-\xd}^{\tr} f_{\XD}(y)\frac{\exp\!\left(-\int_{y + \xd}^{\tr + x} \lambda(\tau) \, \d\tau \right)}{1-\exp(-\Lambda)} \, \d y\Bigg] \label{eq:stationary_cond1}
\end{align}
for $x > \xd$, and
\begin{align}
f_{\XD}(x) &= \lambda(x)\Bigg[\int_0^{\tr + x - \xd} f_{\XD}(y)\frac{\exp\!\left(-\int_{y + \xd}^{\tr + x} \lambda(\tau) \, \d\tau \right)}{1-\exp(-\Lambda)} \, \d y\nonumber\\
&\quad + \int_{\tr + x - \xd}^{\tr} f_{\XD}(y)\frac{\exp\!\left(-\int_{y + \xd}^{2\tr + x} \lambda(\tau) \, \d\tau \right)}{1-\exp(-\Lambda)} \, \d y\Bigg]
\label{eq:stationary_cond2}
\end{align}
for $x \leq \xd$.
In \eqref{eq:stationary_cond1} and \eqref{eq:stationary_cond2}, denote the factors in the brackets as $a(x)$ and we can then write $f_{\XD}(x)=\lambda(x)a(x)$, where $a(x)$ can be interpreted as the attenuation effect on the arrival intensity due to dead time.
It is worth mentioning that similar factorization of $f_{\XD}$ was also used in Isbaner et al. \cite{Isbaner.etal2016} for the derivation of their dead time correction algorithm. However, such a factorization is assumed at the beginning of their derivation, whereas we arrive at this factorization naturally from the stationary condition of a Markov chain. 

Plugging $f_{\XD}(x)=\lambda(x)a(x)$ into \eqref{eq:stationary_cond1}, we have
\begin{align}
a(x) &= \int_0^{x - \xd} \lambda(y)a(y)\frac{\exp\!\left(-\int_{y + \xd}^{x} \lambda(\tau) \, \d\tau \right)}{1-\exp(-\Lambda)} \, \d y \nonumber\\
&\quad + \int_{x - \xd}^{\tr} \lambda(y)a(y)\frac{\exp\!\left(-\int_{y + \xd}^{\tr + x} \lambda(\tau) \, \d\tau \right)}{1-\exp(-\Lambda)} \, \d y.
\label{eq:sx}
\end{align}
Differentiating both sides of the above equation with respect to $x$:
\begin{align}
a'(x)\overset{(a)}{=} & \frac{1}{1-\exp(-\Lambda)}\Bigg[ \lambda(x -\xd) a(x - \xd) \left(1-\exp(-\Lambda)\right)\nonumber\\
& \ -\lambda(x)\bigg(\int_0^{x - \xd} \lambda(y)a(y)\exp\!\left(-\int_{y + \xd}^{x} \lambda(\tau) \, \d\tau \right) \d y \nonumber\\
& \ \ \quad + \int_{x - \xd}^{\tr} \lambda(y)a(y)\exp\!\left(-\int_{y + \xd}^{\tr + x} \lambda(\tau) \, \d\tau \right) \d y \bigg)\Bigg]\nonumber\\
\overset{(b)}{=} & \lambda(x-\xd) a(x-\xd) - \lambda(x)a(x),
\label{eq:diff_s}
\end{align}
where step $(a)$ uses the Leibniz rule and the fact that $\lambda(\tr + x) = \lambda(x)$ and step $(b)$ follows by noticing from \eqref{eq:sx} that the sum of the two integrals equals $(1-\exp(-\Lambda))a(x)$.
Similarly, we can obtain from \eqref{eq:stationary_cond2} that 
\begin{equation}
a'(x) = \lambda(\tr + x-\xd) a(\tr + x - \xd) - \lambda(x)a(x).
\label{eq:diff_sp}
\end{equation}
Note that if we consider periodic extensions of $a(x)$ and $f_{\XD}(x)$, then \eqref{eq:diff_s} and \eqref{eq:diff_sp} are identical. In the following, $a(x)$ and $f_{\XD}(x)$ are considered as their periodic extensions.

Integrating both sides of \eqref{eq:diff_s}, we have that
\begin{equation}
a(x) = - \int_{x - \xd}^x \lambda(\tau)a(\tau) \, \d\tau + C,
\label{eq:delayDiffInt}
\end{equation}
where $C$ is a constant. Multiplying both sides of \eqref{eq:delayDiffInt} by $\lambda(x)$, we have that
\begin{equation}
f_{\XD}(x) = - \lambda(x)\int_{x - \xd}^x f_{\XD}(\tau) \, \d\tau + C\lambda(x).
\label{eq:ftgtC}
\end{equation}
Define
\begin{equation}
g(x):= \int_{x - \xd}^x f_{\XD}(\tau) \, \d\tau.
\label{eq:def_g}
\end{equation}
Since $f_{\XD}(x)$ is a proper probability density function on the interval $[0, \tr)$,  it satisfies
\begin{equation*}
1 = \int_0^{\tr} f_{\XD}(x) \, \d x
  = -\int_{0}^{\tr} \lambda(x)g(x) \, \d x + C\int_{0}^{\tr} \lambda(x) \d x.
\end{equation*}
It follows that
\begin{equation}
C = \frac{1+\int_{0}^{\tr} \lambda(x)g(x) \, \d x}{\int_{0}^{\tr} \lambda(x) \, \d x}=\frac{1+\int_{0}^{\tr} \lambda(x)g(x) \, \d x}{\Lambda}.
\label{eq:C}
\end{equation}
Plugging \eqref{eq:def_g} and \eqref{eq:C} into \eqref{eq:ftgtC}, we have the following relationship between the arrival intensity function $\lambda(x)$ and the limiting distribution of the detection times $f_{\XD}(x)$:
\begin{equation}
f_{\XD}(x) = - \lambda(x)g(x) + \frac{1+\int_{0}^{\tr} \lambda(x)g(x) \, \d x}{\Lambda}\lambda(x).
\label{eq:inverseProb_c}
\end{equation}

\begin{figure*}[t!]
\centering
\includegraphics[width=\textwidth]{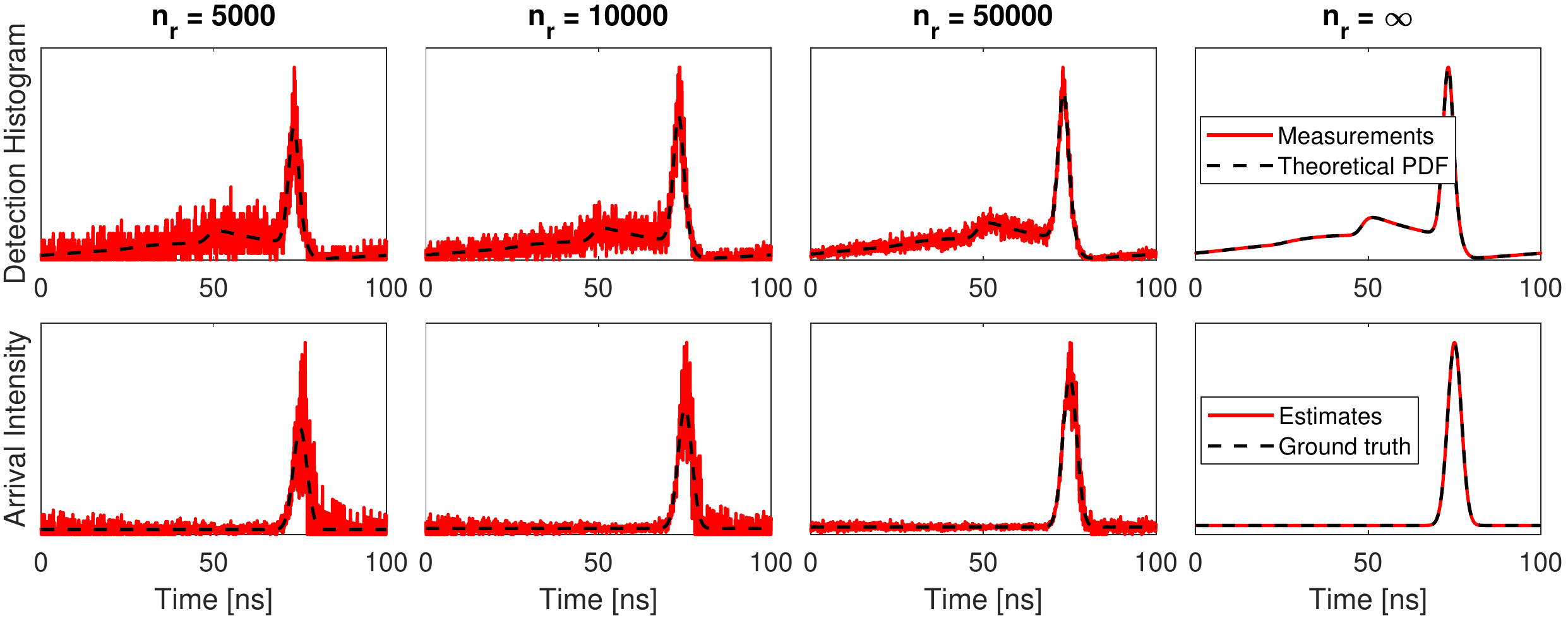}
\caption{Estimation of arrival intensity (bottom) from detection histogram (top) when $S=B=3.16$, $\sigma=2$ ns, $\tbin=50$ ps, $\tr=100$ ns, and $\td=75$ ns. From left to right: increased number of illuminations ($\nr$), where in the last column, theoretical detection histogram is used as measurement.}
\label{fig:hist_knownSB}
\end{figure*}

\subsection{Nonlinear Inverse Formulation and Algorithm}
Suppose that the time interval $[0,\tr)$ is partitioned into $\nb$ equally spaced time bins with bin size $\mathsf{t_{bin}}$, which is the case in TCSPC.\footnote{TCSPC systems digitize photon detection times using time-to-amplitude converters (TACs) or time-to-digital converters TDCs~\cite{Becker2005}. 
The number of bits allocated for an event record is usually fixed, so there is an inherent tradeoff between the acquisition resolution and repetition period. 
For instance, the HydraHarp 400 dedicates 15 bits per record (32\,768 bins) and has a base bin resolution of 1 ps, which can be multiplied by powers of 2~\cite{Wahl2008}.}
Define $\nd := \xd/\mathsf{t_{bin}}$. 
Let the normalized  histogram of detection times be denoted by $\hbf=(h_1,\ldots,h_{\nb})$, where $\sum_{i=1}^\nb h_i=1$. A discrete model for \eqref{eq:inverseProb_c} is then
\begin{equation*}
\hbf = -\diag{\gbf} \lambf + \Lambda^{-1}\lambf + \Lambda^{-1} (\gbf^T\lambf)\lambf + \Delta,
\end{equation*}
where $\lambf=(\lambda_1,\ldots,\lambda_{\nb})$ is a discretization of $\lambda(t)$;
$\Lambda$ is the total flux and is assumed to be known;
$g_i = \sum_{k= i-\nd}^{i-1}h_k$ for $i> \nd$ and $g_i=\sum_{k=i-\nd+\nb}^{\nb}h_k + \sum_{k=1}^{i-1}h_k$ for $i\leq \nd$, which follows from \eqref{eq:def_g};
$\diag{\gbf}$ is a diagonal matrix with $\gbf$ on its diagonal;
and $\Delta$ represents the error due to discretization and the difference between the finite-sample empirical distribution and the limiting distribution. 
For any fixed $\hbf$ (hence fixed $\gbf$), define an operator $\mathcal{T}(\,\cdot\,;\hbf):\mathbb{R}^\nb\to\mathbb{R}^\nb$ as
\begin{equation}
\lambf \mapsto \mathcal{T}(\lambf;\hbf):= -\diag{\gbf} \lambf + \Lambda^{-1}\lambf + \Lambda^{-1} (\gbf^T\lambf)\lambf.
\label{eq:def_T}
\end{equation} 
The inverse problem that we need to solve is then to estimate $\lambf$ from the nonlinear system $\hbf=\mathcal{T}(\lambf;\hbf) + \Delta$ given a measurement vector $\hbf$.
Define the optimization problem:
\begin{equation}
\min_{\lambf} \left\{F(\lambf):= D(\lambf) + \delta_{[0,M]^\nb}(\lambf)\right\},
\label{eq:optimization}
\end{equation}
where $D(\lambf):=\frac{1}{2}\|\hbf - \mathcal{T}(\lambf;\hbf)\|^2$ with $\|\cdot\|$ being the Euclidean norm and $\delta_{[0,M]^\nb}$ the indicator function of the bounded hypercube $[0,M]^\nb$ for some constant $M$. 
Note that while one may include stronger regularizers to reflect prior knowledge about $\lambf$, our goal here is to demonstrate that the proposed method can reconstruct the arrival intensity without any prior knowledge other than the intensity being non-negative; the method is thus 
applicable to a broader class of applications where the arrival intensity is less predictable such as in NLOS imaging.
We use a monotone accelerated proximal gradient (APG) algorithm \cite{Li.Lin2015} to solve \eqref{eq:optimization}. 
Note that the proximal operator for $\delta_{[0,M]^\nb}$ is the orthogonal projector onto $[0,M]^\nb$, denoted by $\Pi_{[0,M]^\nb}(\cdot)$, and the gradient of $D(\lambf)$ is computed as follows:
\begin{align}
&\grad D(\lambf) = \mathbf{J}_{\mathcal{T}}^T\left(\mathcal{T}(\lambf;\hbf) - \hbf\right)\nonumber\\
& = \bigg(\frac{\gbf\lambf^T}{\Lambda} + \frac{1 + \gbf^T\lambf}{\Lambda}\mathbf{I} - \diag{\gbf} \bigg)\left(\mathcal{T}(\lambf;\hbf) - \hbf\right),
\label{eq:grad_nonlinear}
\end{align}
where $\mathbf{J}_\mathcal{T}$ is the Jacobian matrix of $\mathcal{T}$ and $\mathbf{I}$ is the identity matrix. 
We emphasize that $\hbf$, $\gbf$, and $\Lambda$ are fixed throughout the algorithm, thus they are treated as constant instead of functions of $\lambf$ when computing the gradient $\grad D(\lambf)$ in \eqref{eq:grad_nonlinear}.
The convergence of the monotone APG algorithm relies on an appropriate choice of the step size $\gamma$, which should satisfy $\gamma<1/L$, where $L$ is the Lipschitz constant of $\grad D(\cdot)$ \cite{Li.Lin2015}. The following proposition provides an upper-bound $L_u$ for $L$.

\begin{proposition}
\label{prop_converge}
The Lipschitz constant $L$ of the function $\grad D(\cdot)$ defined in \eqref{eq:grad_nonlinear} is upper-bounded by $L_u$ on $[0,M]^\nb$, where $L_u$ is defined as
\begin{equation*}
L_u:=2\Lambda^{-2}\nb M^2 + \left(2\Lambda^{-2} + 2 + 6\Lambda^{-1}\right)\sqrt{\nb}M + 4\Lambda^{-1} + 2.
\end{equation*}
\end{proposition}
\begin{proof}
See Appendix \ref{app_proof_converge}.
\end{proof}

Setting the step size $\gamma=1/L_u$, starting with some initialization $\lambf^0=\lambf^1=\zbf^1\in [0,M]^\nb$ and $q_0=0,q_1=1$, for $k\geq 1$, the monotone APG algorithm for solving \eqref{eq:optimization} proceeds as follows:
\begin{equation}
\begin{split}
\ybf^{k}&=\!\lambf^k \!+\! \frac{q_{k-1}}{q_k}\big(\zbf^k - \lambf^k\big)\!+\!\frac{q_{k-1}-1}{q_{k}}\big(\lambf^{k}-\lambf^{k-1}\big),\\
\zbf^{k+1} &= \Pi_{[0,M]^\nb}\big(\ybf^k - \gamma \grad D(\ybf^k)\big),\\
\xbf^{k+1} &= \Pi_{[0,M]^\nb}\big(\lambf^k - \gamma \grad D(\lambf^k)\big),\\
q_{k+1} &=\frac{\sqrt{4q_k^2+1}+1}{2},\\
\lambf^{k+1}&=\begin{cases}
\zbf^{k+1}, &\text{if } F(\zbf^{k+1})\leq F(\xbf^{k+1}); \\
\xbf^{k+1}, &\text{otherwise}.
\end{cases}
\end{split}
\label{eq:monoAPG}
\end{equation}

Since \eqref{eq:optimization} is a nonconvex optimization problem, a good initialization is important to avoid converging to local minima that are not global minima.
We now introduce an initialization scheme.
Let $C_{\lambda}$ be a scalar that depends on $\lambda$ through $C_{\lambda}=\int_{0}^{\tr} \lambda(x)g(x)\,\d x$. Then \eqref{eq:inverseProb_c} can be written as $f_{\XD}(x)=-\lambda(x)g(x)+\Frac{(1+C_{\lambda})\lambda(x)}{\Lambda}$, which implies
\begin{equation}
\lambda(x)=\frac{f_{\XD}(x)}{\Frac{(1+C_{\lambda})}{\Lambda}-g(x)}.
\label{eq:init_lam}
\end{equation}
Plugging \eqref{eq:init_lam} back into the definition of $C_{\lambda}$, we obtain a fixed point equation for $C_{\lambda}$:
\begin{equation}
C_{\lambda} = \int_{0}^{\tr} \frac{f_{\XD}(x)g(x )}{\Frac{(1+C_{\lambda})}{\Lambda} - g(x)}\, \d x.
\label{eq:init_C}
\end{equation}
Notice that by \eqref{eq:init_lam} the feasible set for $C_{\lambda}$ is $\mathcal{C}=\{C\in\mathbb{R}:\Frac{(1+C)}{\Lambda}-g(x)>0,\forall x\in[0,\tr)\}$. It follows that the RHS of \eqref{eq:init_C} is positive on $\mathcal{C}$ and monotone decreasing to zero as $C_{\lambda}$ goes to infinity. Since the LHS of \eqref{eq:init_C} is linearly increasing, a graph will easily show that \eqref{eq:init_C} has a unique fixed point on $\mathcal{C}$. Therefore, if $f_{\XD}$ is known perfectly, then one can solve \eqref{eq:init_C} for $C_{\lambda}$ and plug $C_{\lambda}$ into \eqref{eq:init_lam} to have a perfect reconstruction of $\lambda$. 
However, in practice, we only have a histogram formed by a limited number of measured detection times. Nevertheless, it is plausible to estimate $C_{\lambda}$ as the fixed point of
\begin{equation}
\widehat{C}_{\lambda} = \sum_{i=1}^{\nb} \frac{h_i\, g_i}{\Frac{(1+\widehat{C}_{\lambda})}{\Lambda} - g_i}.
\label{eq:init_C1}
\end{equation}
We can then use $\lambf^0$ with the $i^{th}$ entry being defined as
\begin{equation}
\lambda^0_i = \frac{h_i}{\Frac{(1+\widehat{C}_{\lambda})}{\Lambda} - g_i}
\label{eq:init_lam1}
\end{equation}
as the initialization of the nonconvex optimization problem. While we do not have a theoretical guarantee for convergence to the global minimum, in our simulation, both $\hbf$ plus random perturbation and the more principled initialization \eqref{eq:init_lam1} lead to good estimates. 
Because solving \eqref{eq:init_C1} is easy and the initial estimate \eqref{eq:init_lam1} is usually close to the solution, using the principled initialization allows the algorithm to converge faster.

Fig.~\ref{fig:hist_knownSB} presents simulated detection histograms and the corresponding arrival intensity estimates using~\eqref{eq:monoAPG}, where $S=B=3.16$ and $\sigma=2$ ns.
We notice that as $\nr$ increases, the detection histogram approaches $f_{\XD}$.
Our estimated arrival intensity likewise approaches the true arrival intensity as $\nr$ increases.
It is interesting to note that while the error in the detection histogram resembles Poisson noise in that the variance increases as the mean increases, the error in the arrival intensity estimate is signal-dependent in a different way.
We observe that the error variance is roughly proportional to the pointwise ratio of $f_{\XD}$ and $\lambda$.
Although we have no theoretical results supporting this hypothesis, the observation suggests that the portions of the arrival intensity easiest to reconstruct are those least attenuated by the dead time effects, and vice versa.

\section{Application to Ranging}
\label{sec:ranging}

We now explore how the theory and algorithm developed in Sections~\ref{sec:formulation} and~\ref{sec:algo} can be used for depth estimation. 
In Section~\ref{subsec:ranging_true}, we assume that the acquisition parameters $S,B,\Lambda$ are known from accurate calibration and compare different methods using the true parameter values. In Section~\ref{subsec:ranging_est}, we provide estimators for $B$ and $S$, and together with the estimator for $\Lambda$ introduced in Section~\ref{subsec:Lam_est}, we test our proposed methods using the estimated parameters.

\subsection{Ranging with True Acquisition Parameters}
\label{subsec:ranging_true}

The ML depth estimator for the Poisson arrival process passes the set of arrivals
through a log-matched filter that is matched to the arrival intensity $\lambda(t)$, where the log-matched filter is defined as
$v(t) := \log(\lambda(t))  = \log\left(f_{\XA}(t)\right) + \log(\Lambda)$ \cite[(33)]{BarDavid1969}.
Given a set of low-flux detection times $\{t_i\}_{i=1}^n$, where photon loss due to dead time is negligible,  the log-matched filter for estimating the depth $z$ is defined as
\begin{equation*}
\begin{split}
\widehat{z} \left(\{t_i\}_{i=1}^n;f_{\XA}\right)
& := \argmax_z \sum_{i=1}^n  \int_{0}^{\tr} \delta(t-t_i)v\left(t+ 2z/c\right) \d t\\
& = \argmax_z \sum_{i=1}^n \log\left(f_{\XA}\left(t_i +2z/c\right)\right).
\end{split}
\end{equation*}

For practical implementation, $\{t_i\}_{i=1}^n$ may be quantized into $\nb$ equally spaced time bins over $[0,\tr)$ with bin centers $\{b_k\}_{k=1}^{\nb}$. A histogram $\mathbf{h}=(h_1,\ldots,h_{\nb})$ for the low-flux detection times can then be obtained from the quantized data $\{\bar{t}_k\}_{k=1}^{n}$. Moreover, instead of estimating the depth $z$, we can estimate the time delay $\tau:=2z/c$, since the mapping from $z$ to $\tau$ is one-to-one.  The estimator for $\tau$ is then
\begin{equation*}
\widehat{\tau} \left(\mathbf{h};f_{\XA}\right) :=  \argmax_{\tau \in \Gamma} \left\{\sum_{k=1}^{\nb} h_k \log\left(f_{\XA}\left(b_k +\tau\right)\right)\right\} + b_{\nb/2},
\end{equation*}
where $\Gamma:=\{-b_{\nb/2},\ldots,-b_1,0,b_1,\ldots,b_{\nb/2}\}$ is a set of on-grid relative time delays, and $f_{\XA}$ is the arrival PDF assuming $b_{\nb/2}$ is the true delay.

In Section~\ref{sec:formulation}, we have derived the limit of the empirical distribution of high-flux detection times $f_{\XD}$. 
Hence, we can similarly define a log-matched filter matched to $f_{\XD}$ and define an estimator as
\begin{equation*}
\widehat{\tau} \left(\mathbf{h};f_{\XD}\right):=  \argmax_{\tau\in\Gamma} \left\{\sum_{k=1}^{\nb} h_k \log\left(f_{\XD}\left(b_k +\tau\right)\right)\right\}+b_{\nb/2},
\end{equation*}
which is preferable if $\hbf$ is obtained via high-flux acquisition with non-negligible photon loss due to dead time.
Note however that $\widehat{\tau} \left(\mathbf{h};f_{\XD}\right)$ is not the ML estimator with dead time effects (even without quantization error), because in this case the joint PDF does not factorize as product of the marginals. While one can obtain the exact joint PDF from the transition PDF \eqref{eq:cond_pdf} and the marginal PDF $f_{\XD}$, the true ML estimator is inconvenient to implement. Therefore, $\widehat{\tau} \left(\mathbf{h};f_{\XD}\right)$ is used in our simulations.

In Section~\ref{sec:algo}, we have derived an algorithm for estimating the arrival time distribution from the detection time distribution. Hence, given a detection histogram, our algorithm can compute an estimate for the arrival histogram $\widehat{\mathbf{h}}^\A$, and then $\tau(\widehat{\hbf}^\A;f_{\XA})$ can be used for depth estimation.

Based on the above discussion, letting $\hbf^{\LF}$ and $\hbf^{\HF}$ denote the detection time histograms obtained via low-flux and high-flux acquisitions, respectively, we compare six depth estimation methods applicable to asynchronous TCSPC systems. 
The methods are as follows:
\begin{enumerate}
\item \textbf{LF}: The low-flux approach first attenuates the incident flux (in practice by applying a neutral density filter) to limit the total flux arriving at the detector to 5\% so that dead time effects can be ignored.
Since the low-flux detection histogram $\hbf^\LF$ can then be considered to be the same as the arrival histogram, it then uses $\widehat{\tau}(\hbf^\LF;f_{\XA})$ as the estimator. 
\item \textbf{HF}: The high-flux method na\"ively assumes that dead time has no effect on the acquisition and uses the estimator $\widehat{\tau}(\hbf^\HF;f_{\XA})$, even when $\hbf^\HF$ is not a good approximation to the arrival histogram. 
\item \textbf{SC}: Shift correction assumes that the dead time only adds a bias to the estimate and that the bias can be computed and subsequently subtracted away.
In practice, this is equivalent to the optical calibration procedure in~\cite{He2013}; for our simulations, we compute the shift in the mode of $f_{\XD}$ compared to that of $f_{\XA}$ and subtract the shift correction from the HF estimate. 
\item \textbf{Isbaner}: This method, based on the work of Isbaner et al., first estimates the arrival histogram $\widehat{\hbf}^\A$ from $\hbf^\HF$ using the algorithm in~\cite{Isbaner.etal2016}, which has publicly available code,\footnote{http://projects.gwdg.de/projects/deadtimecorrectiontcspc} and then applies the estimator $\widehat{\tau}(\widehat{\hbf}^\A;f_{\XA})$. 
While \cite{Isbaner.etal2016} can estimate $\Lambda$ from data, we provide the algorithm with the true $\Lambda$ for fair comparison.
\item \textbf{Proposed method 1 -- MCPDF}: Our first method computes the Markov chain-based PDF $f_{\XD}$ to directly apply $\widehat{\tau}(\hbf^\HF;f_{\XD})$.
\item \textbf{Proposed method 2 -- MCHC}: Our second method is similar to that of Isbaner et al., except it first estimates the arrival histogram $\widehat{\hbf}^\A$ from $\hbf^\HF$ using the Markov chain-based histogram correction algorithm introduced in Section~\ref{sec:algo} and then uses $\widehat{\tau}(\widehat{\hbf}^\A;f_{\XA})$ as the estimator.
\end{enumerate}

We perform Monte Carlo simulations with $\tr = 100$ ns, $\td = 75$ ns, $\sigma = 0.2$ ns, and bin duration 
$\tbin = 5$ ps, which are reasonable experimental parameters for some laboratory settings. 
For each combination of $S$ and $B$, we generate 600 realizations of the arrival process with $\nr = 10^4$ illuminations.
Starting with the first arrival, the high-flux detection sequence is generated by removing subsequent arrivals if they occur within $\td$ of the previous detection.
Generation of the corresponding low-flux detection sequence proceeds in the same manner, but the arrival process is first attenuated via Bernoulli thinning, so photons arrive in only $5\%$ of illumination periods on average.
For each method, the log-matched filtering is performed via circular cross-correlation (circular convolution of the histogram with the time-reversed PDF).
This is due to the asynchronous dead time preserving the shift invariance of the arrival process.

\begin{figure*}[t]
    \centering
    \begin{subfigure}{0.01\linewidth}
        \begin{sideways}
            	\text{$S= 0.1$}
        \end{sideways}
    \end{subfigure}   
    \hfill
    \begin{subfigure}{0.32\linewidth}
        \setlength{\abovecaptionskip}{2pt}
        \centering
        \caption*{$B=0.1$}
        \includegraphics[width=\linewidth]{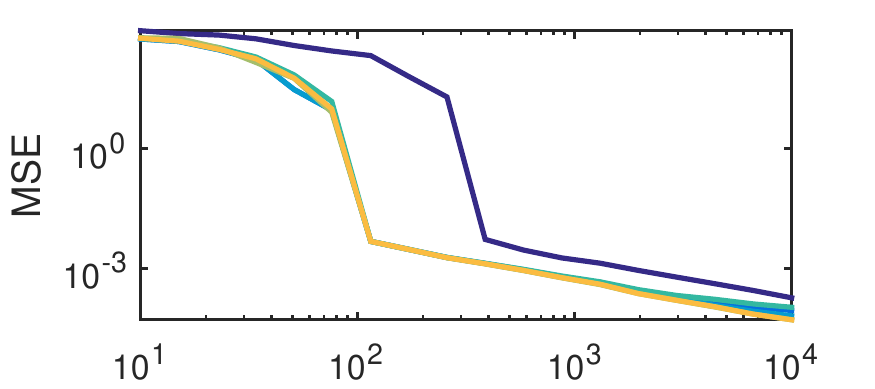}
    \end{subfigure}
    \begin{subfigure}{0.32\linewidth}
        \setlength{\abovecaptionskip}{2pt}
        \centering
        \caption*{$B=0.562$}
        \includegraphics[width=\linewidth]{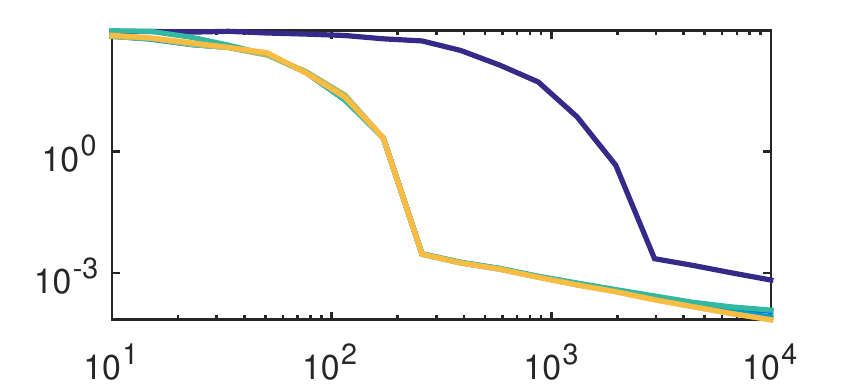}
    \end{subfigure}
    \begin{subfigure}{0.32\linewidth}
        \setlength{\abovecaptionskip}{2pt}
        \centering
        \caption*{$B=3.16$}
        \includegraphics[width=\linewidth]{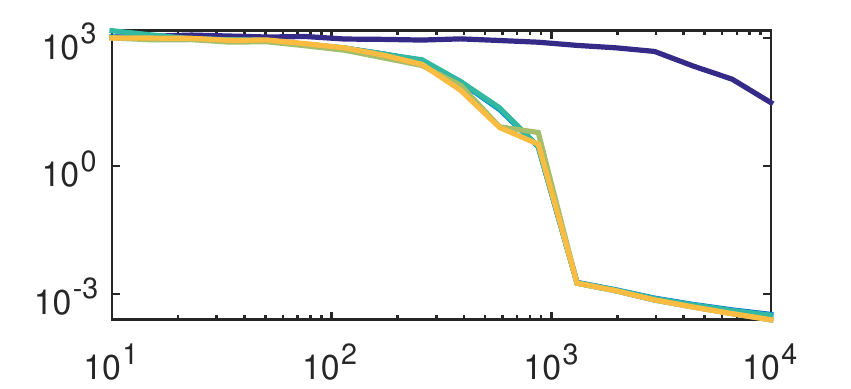}
    \end{subfigure}
    \\
    \begin{subfigure}{0.01\linewidth}
        \begin{sideways}
            	\text{$S= 0.562$}
        \end{sideways}
    \end{subfigure}  
    \hfill
    \begin{subfigure}{0.32\linewidth}
        \centering
        \includegraphics[width=\linewidth]{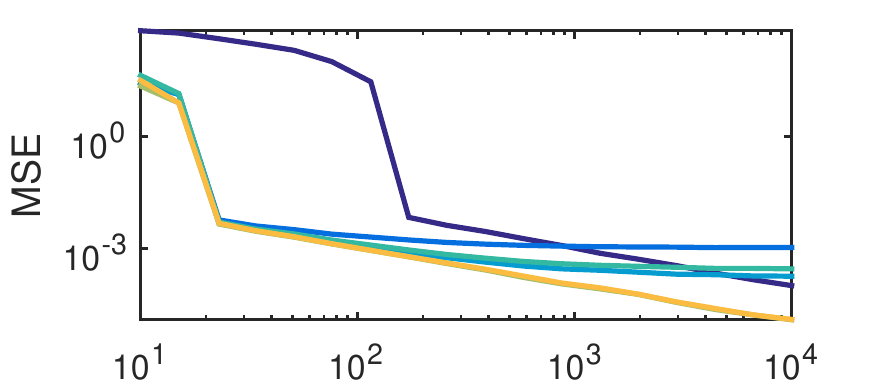}
    \end{subfigure}
        \begin{subfigure}{0.32\linewidth}
        \centering
        \includegraphics[width=\linewidth]{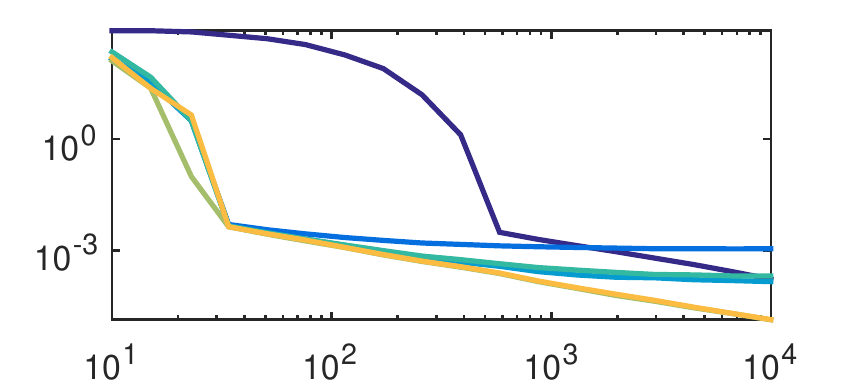}
    \end{subfigure}
    \begin{subfigure}{0.32\linewidth}
        \centering
        \includegraphics[width=\linewidth]{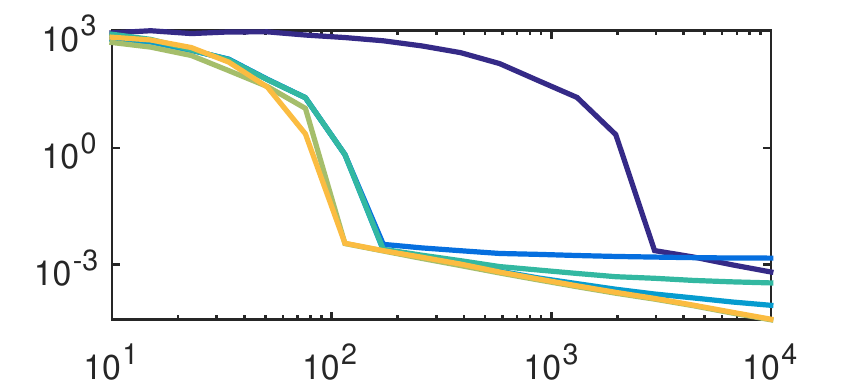}
    \end{subfigure}
    \\
    \begin{subfigure}{0.01\linewidth}
        \begin{sideways}
            	\text{$S= 3.16$}
        \end{sideways}
    \end{subfigure}   
    \hfill
    \begin{subfigure}{0.32\linewidth}
        \centering
        \includegraphics[width=\linewidth]{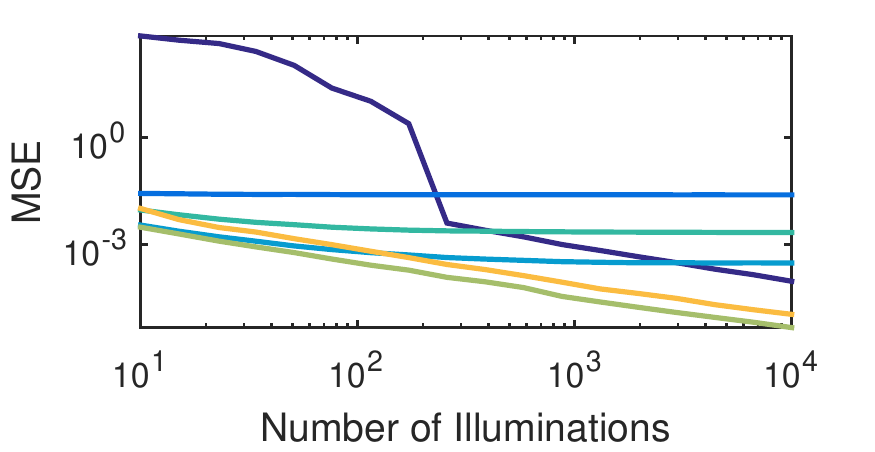}
    \end{subfigure}
        \begin{subfigure}{0.32\linewidth}
        \centering
        \includegraphics[width=\linewidth]{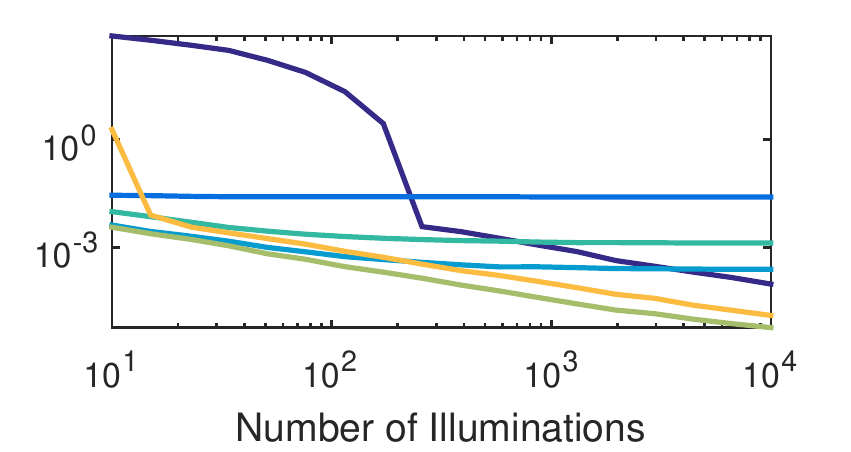}
    \end{subfigure}
    \begin{subfigure}{0.32\linewidth}
        \centering
        \includegraphics[width=\linewidth]{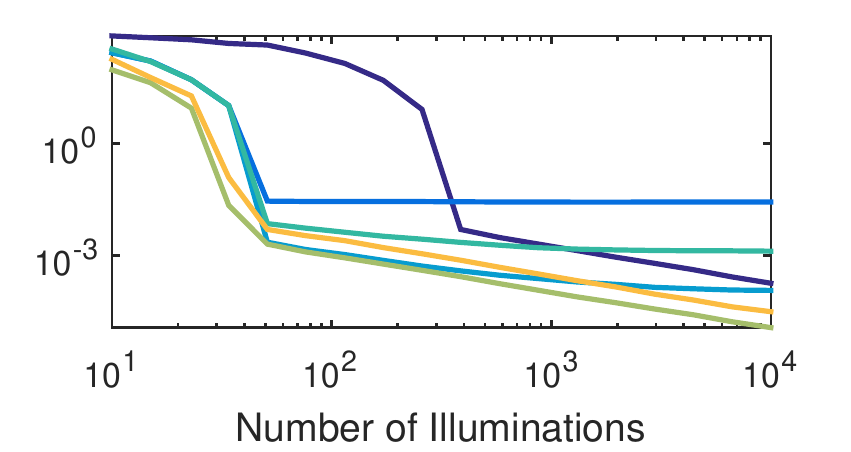}
    \end{subfigure}\\
    \begin{subfigure}{\linewidth}
        \centering
        \includegraphics[width=0.5\linewidth]{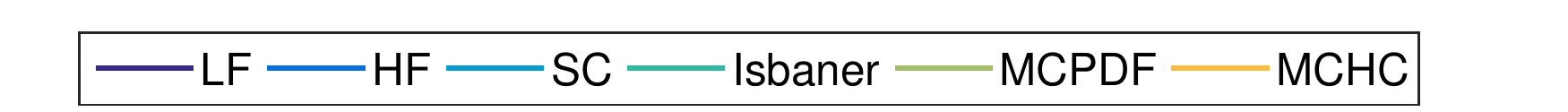}
    \end{subfigure}
    \caption{Plots of the MSE for ranging as a function of $\nr$ for $\tr = 100$ ns, $\td = 75$ ns, $\sigma = 0.2$ ns, $\tbin = 5$ ps and various $S$ and $B$ values. Our proposed methods (MCPDF and MCHC) take advantage of the increased detection rate to perform more accurate ranging than with the low-flux acquisition for all values of $S$, $B$, and $\nr$.}
    \label{fig:MSE_v_numArr}
\end{figure*}

\begin{figure*}[t]
    \centering
    \begin{subfigure}{0.01\linewidth}
        \begin{sideways}
            	\text{$S= 0.1$}
        \end{sideways}
    \end{subfigure}    
    \hfill
    \begin{subfigure}{0.32\linewidth}
        \setlength{\abovecaptionskip}{2pt}
        \centering
        \caption*{$B=0.1$}
        \includegraphics[width=\linewidth]{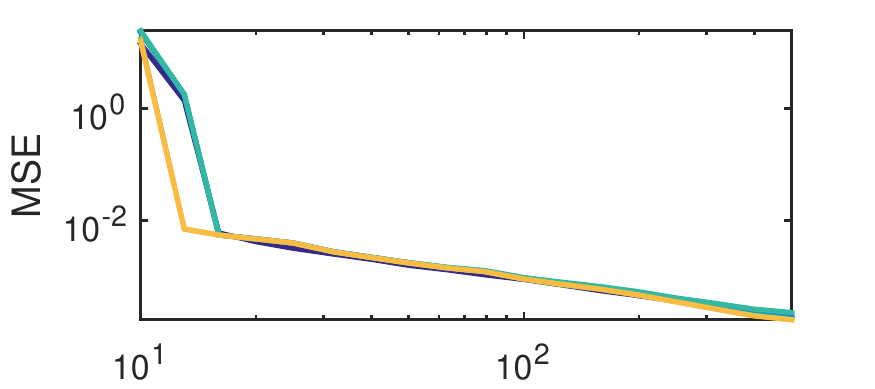}
    \end{subfigure}
    \begin{subfigure}{0.32\linewidth}
        \setlength{\abovecaptionskip}{2pt}
        \centering
        \caption*{$B=0.562$}
        \includegraphics[width=\linewidth]{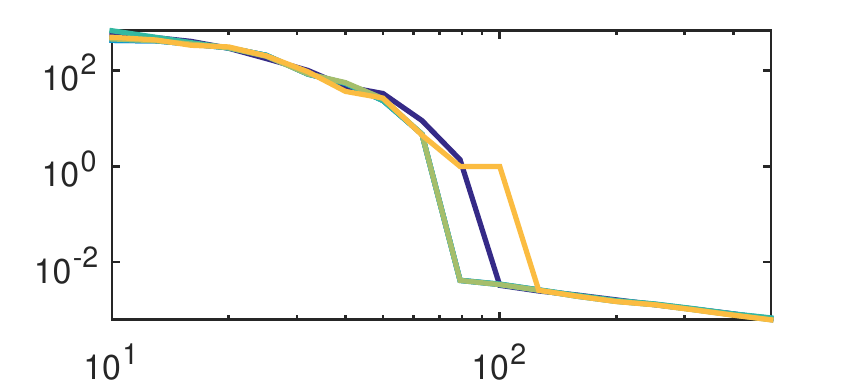}
    \end{subfigure}
    \begin{subfigure}{0.32\linewidth}
        \setlength{\abovecaptionskip}{2pt}
        \centering
        \caption*{$B=3.16$}
        \includegraphics[width=\linewidth]{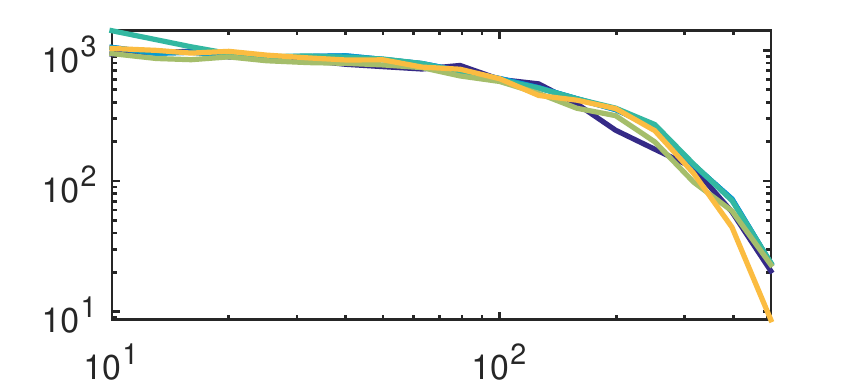}
    \end{subfigure}
    \\
    \begin{subfigure}{0.01\linewidth}
        \begin{sideways}
            	\text{$S= 0.562$}
        \end{sideways}
    \end{subfigure}      
    \hfill
    \begin{subfigure}{0.32\linewidth}
        \centering
        \includegraphics[width=\linewidth]{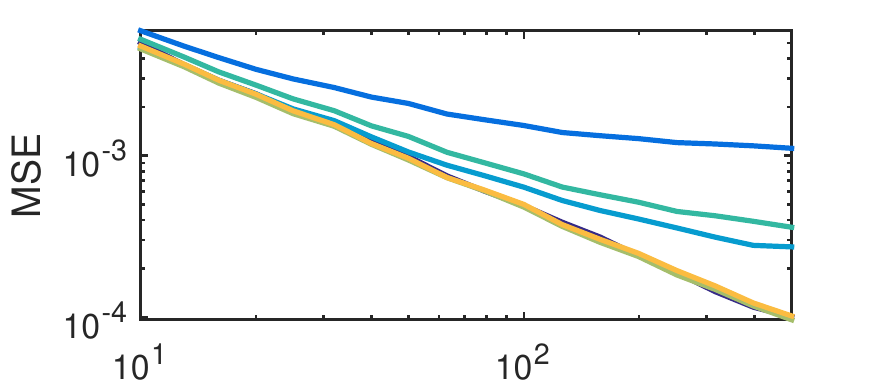}
    \end{subfigure}
        \begin{subfigure}{0.32\linewidth}
        \centering
        \includegraphics[width=\linewidth]{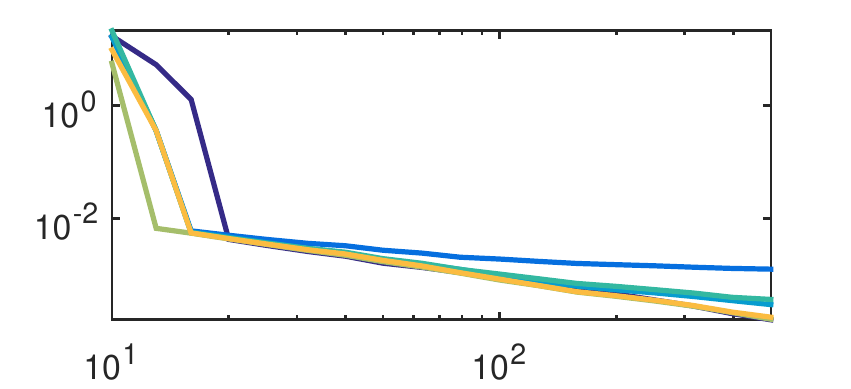}
    \end{subfigure}
    \begin{subfigure}{0.32\linewidth}
        \centering
        \includegraphics[width=\linewidth]{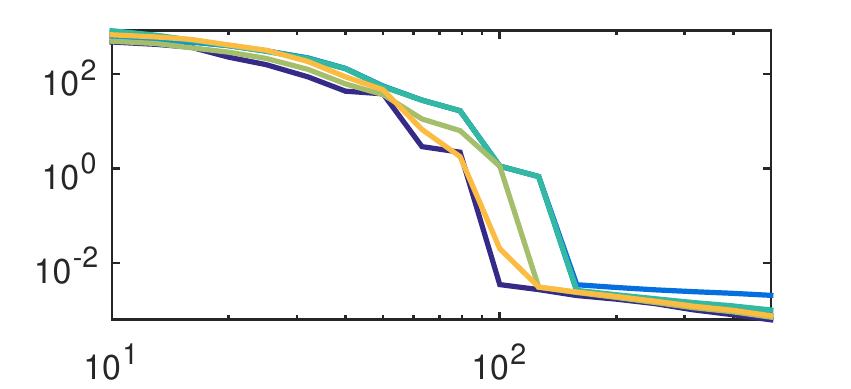}
    \end{subfigure}
    \\
    \begin{subfigure}{0.01\linewidth}
        \begin{sideways}
            	\text{$S= 3.16$}
        \end{sideways}
    \end{subfigure}      
    \hfill
    \begin{subfigure}{0.32\linewidth}
        \centering
        \includegraphics[width=\linewidth]{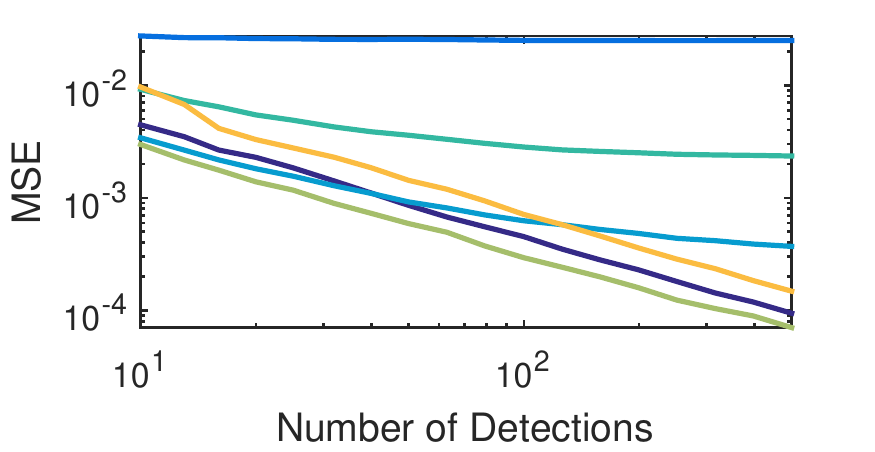}
    \end{subfigure}
        \begin{subfigure}{0.32\linewidth}
        \centering
        \includegraphics[width=\linewidth]{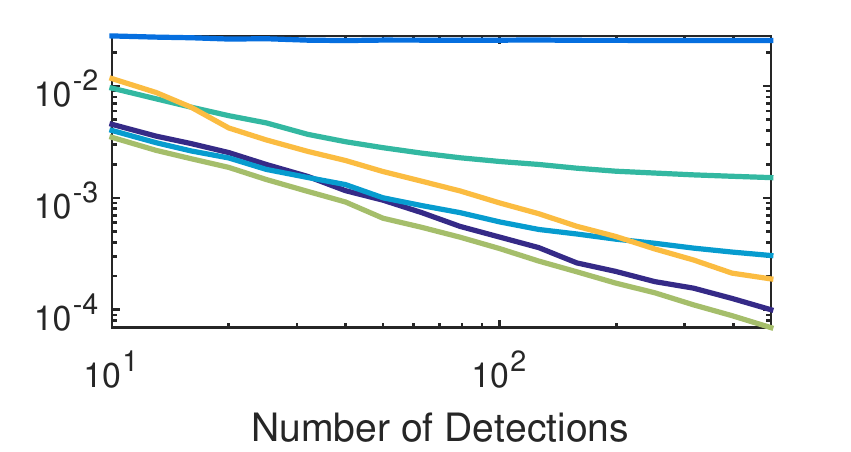}
    \end{subfigure}
    \begin{subfigure}{0.32\linewidth}
        \centering
        \includegraphics[width=\linewidth]{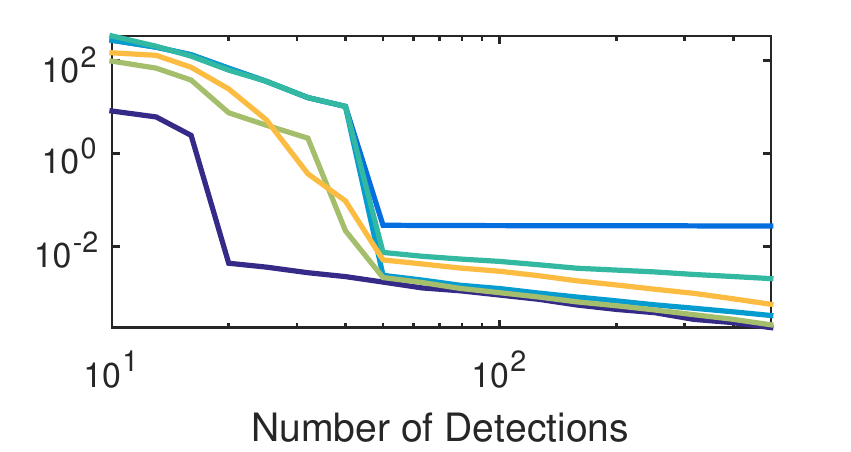}
    \end{subfigure}
    \\
    \begin{subfigure}{\linewidth}
        \centering
        \includegraphics[width=0.5\linewidth]{deadTimeFigs/legend}
    \end{subfigure}
    \caption{Plots of the MSE for ranging as a function of the number of detections for $\tr = 100$ ns, $\td = 75$ ns, $\sigma = 0.2$ ns, $\tbin = 5$ ps and various $S$ and $B$ values. For high $S$ and $B$ not too large, the presence of dead time actually improves ranging accuracy with our MCPDF method relative to the low flux measurements due to the narrowing of the signal pulse.}
    \label{fig:MSE_v_numDet}
\end{figure*}

Fig.~\ref{fig:MSE_v_numArr} compares the MSE for time delay estimation achieved by the six compared methods as a function of the number of illuminations.
We observe that MCPDF usually achieves the lowest MSE, since it directly performs parameter estimation with the updated detection model. 
Equivalently, MCPDF needs the fewest illuminations to achieve a given MSE, hence enabling the fastest acquisition.
The MSE of MCHC is comparable to that of MCPDF, limited only in that it must first invert the histogram before estimating the depth.
Compared to the LF approach, both MCPDF and MCHC require fewer illuminations to achieve the same MSE, and that time efficiency increases as $S$ and $B$ increase and dead time has a more significant impact.
Regarding the other approaches, HF is more effective only for low numbers of illuminations but the estimate quickly becomes biased and is therefore not suited to precise depth measurement.
Correcting for this bias with SC is quite effective for extending to somewhat higher $\nr$, although eventually, more accurate modeling is necessary for more precise estimates.
The state-of-the-art method for dealing with asynchronous dead-time models by Isbaner et al.~\cite{Isbaner.etal2016}  achieves low MSE when the total flux is low or moderate, while the accuracy degrades in high-flux scenarios. The performance degradation is due to their approximation of the detection time distribution being less accurate in high-flux settings (S.~Isbaner, personal communication, May 14, 2018).

In addition to enabling faster acquisition, we explored whether dead time could lead to more accurate ranging for an equal number of detected photons. 
The Fisher information analysis in Section~\ref{subsec:fisher} has provided a theoretical prediction that for sufficiently high SBR, estimating depth from the dead time-distorted detection time distribution can yield lower MSE than that from the arrival time distribution. 
Although the estimators in our Monte Carlo simulation are not guaranteed to achieve the Cram{\'e}r-Rao lower bound (i.e., the reciprocal of Fisher information), we would like to see whether the reduction of ranging error due to dead time 
also exists with simple and commonly used estimators. 
Fig.~\ref{fig:MSE_v_numDet} compares the MSE for time delay estimation by the six methods as a function of the number of detections.
We notice that for the high SBR cases where $S=3.16, B=0.1$ and $S=3.16, B=0.562$, MCPDF outperforms LF, which provides numerical evidence that dead time can be beneficial when properly modeled.

\subsection{Ranging with Estimated Acquisition Parameters}
\label{subsec:ranging_est}
The results in Section~\ref{subsec:ranging_true} use methods that compute $f_{\XD}$, $f_{\XA}$, and $\widehat{\hbf}^\A$ assuming the true values of $B$, $S$, and $\Lambda$ are known.
However, in most practical scenarios, this information will not be available \emph{a priori}.
We describe here some strategies that can be used for determining those parameters.

\subsubsection{Maximum Likelihood Estimator for $B$}
We begin by assuming that background calibration measurements can occasionally be made within the ranging process, for which photons are detected while the laser is turned off. 
For sequences of ranging measurements such as in 3D imaging, such background-only acquisitions could be made for each laser position or for sets of laser points (e.g., once per row or once per image).
Since the background process is homogeneous with $\lambda(t) = \lambdab$, we can rewrite~\eqref{eq:con_density_T} as
\begin{equation*}
f_{T_{i+1}|T_i}(t|t_i)=\lambdab\exp\big(-\lambdab\big(t-(t_i+\td)\big)\big)\mathbb{I}\{t>t_i+\td\}.   
\end{equation*}
Then the conditional distribution of $T_i$'s given $T_1=t_1$ is
\begin{equation*}
\begin{split}
    &f_{T_2,\ldots,T_n|T_1}(t_2,\ldots,t_n|t_1)=\prod_{i=1}^{n-1}f_{T_{i+1}|T_i}(t_{i+1}|t_i)\\
    &=\lam_b^{n-1}\exp(-\lambdab (t_n-t_1) + (n-1)\lambdab \td).
\end{split}
\end{equation*}
Given a set of absolute detection times $\{t_i\}_{i=1}^n$, the (conditional) log-likelihood function $\ln\left( f_{T_2,\ldots,T_n|T_1}\right)$ is
\begin{equation*}
\mathcal{L}(\{t_{i}\}_{i=1}^n;\lambdab)=(n-1)\ln(\lambdab) - \lambdab (t_n-t_1) + (n-1)\lambdab\td.
\end{equation*}
Setting the derivative of $\mathcal{L}(\{t_{i}\}_{i=1}^n;\lambdab)$ with respect to $\lambdab$ to zero, we obtain the (conditional) ML estimator for $\lambdab$ as 
\begin{equation*}
 \lambdahatb^{\text{ML}} = \frac{n-1}{(t_n-t_1) - (n-1)\td}.
\end{equation*}
It follows that $\Bml = \lambdahatb^{\mathrm{ML}}\tr$.
\begin{figure}[t]
    \centering
    \begin{subfigure}{0.49\linewidth}
        \centering
        \includegraphics[trim={0 0 4mm 2mm},clip,width=\linewidth]{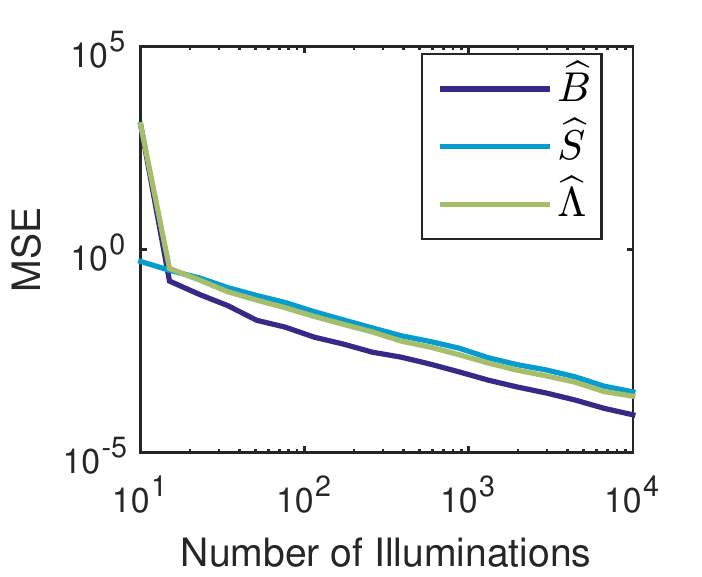}
        \caption{}
        \label{fig:BSL_est}
    \end{subfigure} 
    \hfill
    \begin{subfigure}{0.49\linewidth}
        \centering
        \includegraphics[trim={0 0 4mm 2mm},clip,width=\linewidth]{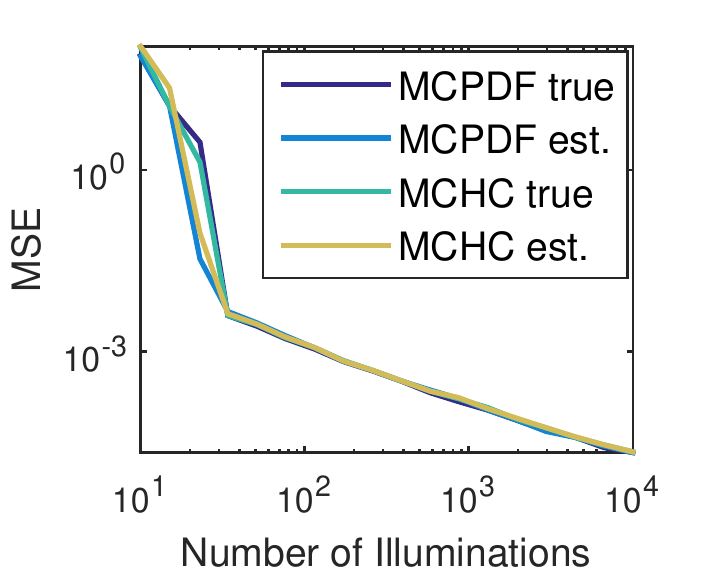}
        \caption{}
        \label{fig:Z_est}
    \end{subfigure}
    \caption{Plots of the parameter estimates as a function of $\nr$ for $S=B=0.562$, $\tr = 100$ ns, $\td = 75$ ns, $\sigma = 0.2$ ns, and $\tbin = 10$ ps.
    In (a), the estimates $\Bhat$, $\Lhat$, and $\Shat$ improve as $\nr$ increases.
    The ranging results in (b) using estimated parameters show no degradation in performance compared to the methods with parameters known \emph{a priori}.}
    \label{fig:est_parameters}
\end{figure}

\subsubsection{Estimating $S$}
From $\Lml$ and $\Bml$, we could also compute $\Shat = \max(\Lml-\Bml,0)$ to ensure non-negativity.
However, setting $\Shat=0$ whenever $\Bml>\Lml$ is not informative for depth estimation with the log-matched filter, since $f_{\XD}$ and $f_{\XA}$ would be uniform PDFs.
Instead, one can assume that there is always at least some small amount of signal and background in the ranging process, so we choose to set minimum values of $\Smin = \Bmin = 0.01$.
Then $\Bhat=\max(\Bml,\Bmin)$, $\Lhat=\max(\Lml,\Bhat+\Smin)$, and $\Shat=\Lhat-\Bhat$.

Fig.~\ref{fig:est_parameters} shows one example of estimates using this strategy for 500 Monte Carlo trials with $S=B=0.562$, $\tr = 100$ ns, $\td = 75$ ns, $\sigma = 0.2$ ns, and $\tbin = 10$ ps.
In Fig.~\ref{fig:BSL_est}, the $\Bhat$, $\Lhat$, and $\Shat$ estimates consistently improve as $\nr$ increases beyond a very small number of detections.
The resulting depth estimates shown in Fig.~\ref{fig:Z_est} are virtually indistinguishable from the methods using the known parameter values.
It is worth noting that, while $\Bml$ applies for any value of $B$, $\Lml$ becomes less reliable for large $\Lambda$ since $P(R=0)$ in~\eqref{eq:K_pdf} approaches unity.
If the number of illuminations was not fixed in advance, one could pursue an adaptive acquisition strategy as in~\cite{Medin2018}.
Alternatively, for 3D imaging, one could take advantage of spatial correlations to estimate $\Lambda$, for example, to borrow measurements from neighboring pixels.

\section{Conclusion}

This work studied dead time compensation for a modern, asynchronous, nonparalyzable detector. 
By using a Markov chain model for detection times, we obtained the limit of the empirical distribution of detection times as the stationary distribution of the Markov chain. 
We found that the Fisher information per detection can be higher for this limiting distribution than for the arrival distribution, which suggests that the distortion due to dead time 
can be beneficial for depth estimation if used properly. 
Indeed, simulation results showed that our first proposed method MCPDF, which is a log-matched filter matched to the limiting distribution, achieved lower error than the low-flux method for a fixed number of detections when the SBR is sufficiently high. 
By exploiting the stationary condition for the Markov chain, we derived our second proposed method MCHC, which estimates the arrival distribution from the detection distribution by solving a nonlinear inverse problem with a provably convergent optimization algorithm, and then the corrected histogram is used in a log-matched filter estimator. 
Although we only tested MCHC in the context of ranging, it makes no assumptions about the arrival intensity and should thus be applicable to other dead-time-limited TCSPC applications, including FLIM and NLOS imaging; we leave such extensions as future work. 

\appendices
\section{Proof of Proposition \ref{prop_markov}}
\label{app_proof_markov}
First, we show that $\{X_i\}_{i\in\mathbb{N}}$ is a Markov chain. Define $\mathcal{B}_i:=\{X_{i}\leq x_{i}\},\forall i\in\mathbb{N}$.  We need to establish that
\begin{equation}
P\left(\mathcal{B}_{i+1}\vert X_k = x_k,\forall k\leq i \right) = P\left(\mathcal{B}_{i+1}\vert X_i=x_i \right).
\label{eq:markov_property}
\end{equation} 
The following equivalence of events will be useful in the proof, as it relates the sets defined by elements of $\{X_i\}_{i\in\mathbb{N}}$ to those of $\{T_i\}_{i\in\mathbb{N}}$, which has a known transition density \eqref{eq:con_density_T}:
\begin{equation}
\begin{split}
\{ X_{i+1}\leq x_{i+1}\} &= \cup_{k=0}^{\infty} \{k\tr < T_{i+1} \leq k\tr + x_{i+1}\}, \\
\{T_i = k_i \tr + x_i\}&=\{K_i= k_i \}\cap \{X_i = x_i\}.
\end{split}
\label{eq:id_X_T}
\end{equation}
Let $\td=\kd\tr + \xd$, where $\kd=\lfloor\td/\tr \rfloor$ and $\xd=\td\mod \tr$. Moreover, define $\mathcal{A}_k:= \{k\tr < T_{i+1} \leq k\tr + x_{i+1}\}, \forall k\in \mathbb{N}\cup\{0\}$. Then we have 
\begin{align}
&P\left(\mathcal{B}_{i+1}  \vert X_j = x_j,\forall j\leq i \right)\nonumber\\
&=\sum_{k_1,...,k_i} \Bigg( P\left( \mathcal{B}_{i+1}  \vert X_j = x_j, K_j=k_j, \forall j\leq i \right)\nonumber\\
&\qquad\qquad\cdot P(K_{j}=k_j,\forall j\leq i \vert X_j=x_j,\forall j\leq i)\Bigg)\nonumber\\
&= \sum_{k_1,...,k_i}\Bigg( P\left(\cup_{k=0}^{\infty}\mathcal{A}_k \vert T_j=k_j \tr + x_j,\forall j\leq i\right)\nonumber\\
&\qquad\qquad \cdot P(K_{j}=k_j,\forall j\leq i \vert X_j=x_j,\forall j\leq i)\Bigg)\label{eq:markov1}
\end{align}
where the summation is over all $0\leq k_1\leq \cdots \leq k_i < \infty$ and
the last equality follows by \eqref{eq:id_X_T}.
In the following, we show that the first probability in \eqref{eq:markov1} only depends on $x_i$ and $\xd$.
\begin{align}
&P\left(\cup_{k=0}^{\infty}\mathcal{A}_k \vert T_j=k_j \tr + x_j,\forall j\leq i\right)\nonumber\\
&\overset{(a)}{=}\sum_{k=0}^{\infty} P\left(\mathcal{A}_k \vert T_i=k_i \tr + x_i\right)\nonumber\\
&\overset{(b)}{=} \int_{(k_i+\kd)\tr + x_i + x_d}^{(k_i+\kd)\tr + x_{i+1}} \!\!\lambda(t) \exp\left(-\int_{(k_i+\kd)\tr + x_i+\xd}^{t} \lambda(\tau)\d \tau\right)\d t\nonumber\\
&+\!\!\!\sum_{k=k_i+\kd+1}^{\infty} \int_{k\tr}^{k\tr + x_{i+1}} \!\!\!\!\lambda(t) \exp\left(-\int_{(k_i+\kd)\tr + x_i+\xd}^{t} \!\!\!\lambda(\tau)\d \tau\right)\d t,
\label{eq:S1S2}
\end{align}
where step $(a)$ follows by the Markov property of $\{T_i\}_{i=1}^{\infty}$ and $\{\mathcal{A}_k\}$ being disjoint and in step $(b)$, we have plugged in \eqref{eq:con_density_T} and assumed that $x_i + \xd \leq x_{i+1}\leq \tr$.
Note that other relationships between $x_i$, $x_{i+1}$, $\xd$, $\tr$ may lead to slightly different expression, but the derivation follows similarly. (We will see that the expression does not depend on $\kd$.)
Label the two terms in \eqref{eq:S1S2} as $S_1$ and $S_2$. First, consider $S_1$:
\begin{align*}
S_1 &= \int_{x_i + \xd}^{x_{i+1}} \lambda(t)\exp\!\left(-\int_{k_i \tr + x_i + \xd}^{k_i \tr + t} \lambda(\tau)\d\tau\right) \d t\\
&= \int_{x_i + \xd}^{x_{i+1}} \lambda(t)\exp\!\left(-\int_{x_i + \xd}^{ t} \lambda(\tau)\d\tau\right) \d t ,
\end{align*} 
which follows by change of variable and  $\lambda(t+k_i \tr) = \lambda(t)$.
Next consider $S_2$:
\begin{align*}
S_2 &=\!\!\!\!\!\!\sum_{k=k_i+\kd+1}^{\infty} \int_{0}^{x_{i+1}} \lambda(t)\exp\!\left(-\int_{(k_i+\kd) \tr + x_i + \xd}^{k \tr + t} \lambda(\tau)\d\tau\right)\d t\\
&=\sum_{k=0}^{\infty} \int_{0}^{x_{i+1}} \lambda(t)\exp\!\left(-\int_{(k_i+\kd) \tr + x_i + \xd}^{(k+k_i+\kd+1)\tr + t} \lambda(\tau) \, \d\tau\right)\d t\\
&=\sum_{k=0}^{\infty} \left(\exp(-\Lambda)\right)^k \int_{0}^{x_{i+1}} \lambda(t)\exp\!\left(-\int_{x_i + \xd}^{\tr + t} \lambda(\tau) \, \d\tau\right)\d t\\
&= \frac{\int_{0}^{x_{i+1}} \lambda(t)\exp\!\left(-\int_{x_i + \xd}^{\tr + t} \lambda(\tau) \, \d\tau\right)\d t}{1-\exp(-\Lambda)} .
\end{align*}
Notice that neither $S_1$ nor $S_2$ depends on $\kd$, $\{K_j\}_{j\leq i}$, or $\{X_j\}_{j<i}$. 
Plugging $S_1$ and $S_2$ back into \eqref{eq:markov1}, we have that
\begin{align*}
&P(\mathcal{B}_{i+1}\vert X_j=x_j,\forall j\leq i)\\
&=\sum_{k_1,\ldots,k_i}(S_1 + S_2)P(K_{j}=k_j,\forall j\leq i \vert X_j=x_j,\forall j\leq i)\\
&=S_1 + S_2,
\end{align*}
where the last equality holds since $P(\cdot\vert X_j=x_j,\forall j\leq i)$ is a probability measure and that the summation $\sum_{k_1,\ldots,k_i}$ is over all $0\leq k_1\leq \ldots\leq k_i<\infty$.
Hence, we have established \eqref{eq:markov_property}, and therefore proved that $\{X_i\}_{i\in\mathbb{N}}$ is a Markov chain. 

Next, we compute the transition PDF to justify \eqref{eq:cond_pdf}:
\begin{align*}
&f_{X_{i+1}|X_i}(x_{i+1}|x_i)= \frac{\mathsf{d}}{\mathsf{d}x_{i+1}} P(X_{i+1}\leq x_{i+1}\vert X_{i}= x_i)\\
& = \frac{\mathsf{d}}{\mathsf{d}x_{i+1}} S_1 + \frac{\mathsf{d}}{\mathsf{d}x_{i+1}} S_2 = \frac{\lambda(x_{i+1})\exp\!\left(-\int_{x_i + \xd}^{ x_{i+1}} \lambda(\tau) \, \d\tau \right)}{1-\exp(-\Lambda)}.
\end{align*}
Recall that we have assumed $x_i + \xd \leq x_{i+1}\leq \tr$ in the derivation above, and we can check that it matches \eqref{eq:cond_pdf} for this case. Other cases can be derived similarly.

\section{Derivation of Fisher Information}
\label{app_fisher_info}
We present the derivation for $\FID$; the derivation for $\FIA$ follows similarly. By definition of Fisher information:
\begin{align*}
&\FID = \int_{0}^{\tr} \left(- \frac{\partial^2}{\partial z^2} \log \left(f_{\XD}(x;z)\right)\right) f_{\XD}(x;z) \, \d x\\
& = \int_{0}^{\tr}  \frac{\left(\frac{\partial}{\partial z}f_{\XD}(x;z)\right)^2 }{f_{\XD}(x;z)} \,\d x - \int_{0}^{\tr} \frac{\partial^2}{\partial z^2} f_{\XD}(x;z) \,\d x\\
& \overset{(a)}{=} \int_{0}^{\tr}  \frac{\left(\frac{\partial}{\partial z}f_{\XD}(x;z)\right)^2 }{f_{\XD}(x;z)} \,\d x - \frac{\partial^2}{\partial z^2}\left(\int_{0}^{\tr}  f_{\XD}(x;z) \,\d x\right)\\
&=\int_{0}^{\tr}  \frac{\left(\frac{\partial}{\partial z}f_{\XD}(x;z)\right)^2 }{f_{\XD}(x;z)} \,\d x,
\end{align*}
where the interchange of derivative and integral in step $(a)$ holds trivially, since the range of the integral is finite.

\section{Proof of Proposition \ref{prop:K_iid}}
\label{app_proof_K_iid}

In the following, we will show that
\begin{equation*}
    P\Big(\!R_i=r_i,R_{i-1}=r_{i-1}\!\Big)\! =\! \Big(1-\exp(-\Lambda)\Big)^2\!\!\!\prod_{j=i-1}^i\!\!\exp(-r_{j}\Lambda),
\end{equation*}
which would imply that Proposition~\ref{prop:K_iid} is valid for $R_i$ and $R_{i-1}$; the proof for more than two $R_i$'s follows similarly. 

Define event $\mathcal{E}_j$ for $j=i-1,i$ as
\begin{equation*}
    \mathcal{E}_j:=\Big\{r_j\tr+T_{j}+\td\leq T_{j+1}<(r_j+1)\tr+T_{j}+\td\Big\}.
\end{equation*}
By definition of $R_i$ in \eqref{eq:K_def},
\begin{align*}
    P\Big(R_{i}=r_{i},R_{i-1}=r_{i-1}\Big)&=P\Big(\mathcal{E}_{i}\cap\mathcal{E}_{i-1}\Big)\\
    &=\mathbb{E}\Big[P\Big(\mathcal{E}_{i}\cap\mathcal{E}_{i-1}\,\Big\vert\, T_{i-1}\Big)\Big].
\end{align*}
Note that by the Markov property of absolute detection times as discussed in Section~\ref{subsec:markov},  the joint PDF of $T_{i+1},T_i$ given $T_{i-1}=t_{i-1}$ is
\begin{align*}
    f_{T_{i+1},T_{i}|T_{i-1}}(t_{i+1},t_i|t_{i-1})=f_{T_{i+1}|T_i}(t_{i+1}|t_i)f_{T_{i}|T_{i-1}}(t_{t}|t_{i-1}).
\end{align*}
Let $a_j:=r_j\tr+t_j+\td$ for $j=i-1,i$. Then
\begin{align*}
    &P\Big(\mathcal{E}_1,\mathcal{E}_2\,\Big\vert\,T_{i-1}=t_{i-1}\Big)\\
    &= \int_{a_{i-1}}^{\tr+a_{i-1}} \int_{a_i}^{\tr+a_i}\!\! f_{T_{i+1}|T_i}(t_{i+1}|t_i)f_{T_{i}|T_{i-1}}(t_{i}|t_{i-1}) \,\d t_{i+1} \,\d t_{i}.
\end{align*}
First consider the inner integral:
\begin{align*}
    &\int_{a_i}^{\tr+a_i} f_{T_{i+1}|T_i}(t_{i+1}|t_i) \,\d t_{i+1}\\
    &=\int_{a_i}^{\tr+a_i} \lambda(t_{i+1})\exp\Big(-\int_{t_i+\td}^{t_{i+1}}\lambda(\tau)\,\d\tau\Big)\,\d t_{i+1}\\
    &= -\exp\Big(-\int_{t_i+\td}^{t_{i+1}}\lambda(\tau)\,\d\tau\Big)\,\Big\vert_{t_{i+1}=a_i}^{t_{i+1}=\tr+a_i}\\
    &=(1-\exp(-\Lambda))\exp(-r_i\Lambda).
\end{align*}
Note that the inner integral does not depend on $t_i$. Using similar calculation, we have that the outer integral does not depend on $t_{i-1}$. Hence, 
\begin{align*}
    &P\Big(\mathcal{E}_1\cap\mathcal{E}_2\,\Big\vert\, T_{i-1}=t_{i-1}\Big) =P\Big(\mathcal{E}_1\cap\mathcal{E}_2\Big) \\
    &\qquad\qquad=\Big(1-\exp(-\Lambda)\Big)^2\prod_{j=i-1}^i\exp(-r_{j}\Lambda),
\end{align*}
which is the desired result.

\section{Proof of Proposition \ref{prop_converge}}
\label{app_proof_converge}

In the following, we will find an upper bound for the Lipschitz constant $L$ of $\grad D(\cdot)$ defined in \eqref{eq:grad_nonlinear}.
For brevity, we omit the dependence on $\hbf$ in the notation for $\mathcal{T}$. By \eqref{eq:grad_nonlinear}, we have
\begin{equation*}
\begin{split}
\Lambda \nabla D(\lambf) &= \gbf\lambf^T \mathcal{T}(\lambf) + \mathcal{T}(\lambf) + \gbf^T\lambf \mathcal{T}(\lambf) - \Lambda\,\diag{\gbf}\mathcal{T}(\lambf)\\
& \quad - \gbf\lambf^T \hbf - \hbf - \gbf^T\lambf \hbf + \Lambda\,\diag{\gbf}\hbf.
\end{split}
\end{equation*}
It follows that for any $\ubf,\vbf\in [0,M]^\nb$, we have by triangle inequality that
\begin{align*}
\Lambda & \|\nabla D(\ubf) - \nabla D(\vbf)\| \\
& \leq \|\gbf\ubf^T \mathcal{T}(\ubf)- \gbf\vbf^T \mathcal{T}(\vbf)\|
  + \|\mathcal{T}(\ubf)-\mathcal{T}(\vbf)\| \\
& + \|\gbf^T\ubf \mathcal{T}(\ubf)-\gbf^T\vbf \mathcal{T}(\vbf)\|\\
& + \Lambda\|\diag{\gbf}\mathcal{T}(\ubf)-\diag{\gbf}\mathcal{T}(\vbf)\| \\
& + \|\gbf\ubf^T \hbf - \gbf\vbf^T \hbf\| 
  + \|\gbf^T\ubf \hbf - \gbf^T\vbf \hbf\|.
\end{align*}
Label the six terms on the right hand side as $T_1,\ldots,T_6$.
We will show that there exist constants $L_1,\ldots, L_6<\infty$ such that $T_i\leq L_i\|\ubf-\vbf\|, \forall i=1,\ldots, 6$.
Then the Lipschitz constant $L$ of the gradient $\nabla D$ is upper bounded by $\Lambda^{-1}\sum_{i=1}^6 L_i$.

First consider $T_2$. Let $\hat{g}:=\max_{i\in[\nb]}g_i$. Then
\begin{align*}
&T_2 \overset{(a)}{\leq} \|\gbf^T\ubf\ubf - \gbf^T\vbf\vbf\| + \|\ubf-\vbf\| + \|\diag{\gbf}(\ubf - \vbf)\|\\
&\overset{(b)}{\leq} \|\gbf^T\ubf\ubf - \gbf^T\ubf\vbf\|+ \|\gbf^T\ubf\vbf - \gbf^T\vbf\vbf\| + (1+\hat{g})\|\ubf-\vbf\|\\
&\overset{(c)}{\leq} \|\gbf\|\|\ubf\|\|\ubf-\vbf\| + \|\gbf\|\|\vbf\|\|\ubf-\vbf\|+(1+\hat{g}) \|\ubf-\vbf\|\\
&\overset{(d)}{\leq} 2\sqrt{\nb}M \|\ubf-\vbf\| + 2\|\ubf-\vbf\| = 2(\sqrt{\nb}M + 1)\|\ubf-\vbf\|,
\end{align*}
where step $(a)$ follows by triangle inequality, step $(b)$ follows by triangle inequality and the fact that the largest eigenvalue of a diagonal matrix equals to the largest entry on its diagonal, and step $(c)$ follows by Cauchy--Schwarz. To see step $(d)$, notice that $\|\ubf\|,\|\vbf\| \leq \sqrt{\nb}M$, since $\ubf,\vbf \in [0,M]^\nb$ and $\max_{i\in[\nb]}g_i\leq\|\gbf\|\leq \|\gbf\|_1 \leq \|\hbf\|_1=1$ (the second inequality follows by the fact that $\gbf$ is non-negative and so $\|\gbf\|=\sqrt{\sum_{i=1}^\nd g_i^2} \leq \sqrt{(\sum_{i=1}^\nd g_i)^2}=\|\gbf\|_1$, the third inequality assumed $\td\leq \tr$ and the last equality follows by $\hbf$ being a proper probability density function). Similarly, we can show that 
$T_4 \leq 2\Lambda(\sqrt{\nb}M + 1)\|\ubf-\vbf\|$, $T_5 \leq \|\ubf-\vbf\|$, and $T_6\leq \|\ubf-\vbf\|$.

Next consider $T_1$:
\begin{align*}
T_1 & \overset{(a)}{\leq}\left\vert \ubf^T \mathcal{T}(\ubf)-\vbf^T \mathcal{T}(\vbf)\right\vert\\
& \overset{(b)}{\leq} \left\vert \ubf^T \mathcal{T}(\ubf)-\vbf^T \mathcal{T}(\ubf)\right\vert + \left\vert \vbf^T \mathcal{T}(\ubf)-\vbf^T \mathcal{T}(\vbf)\right\vert\\
&\overset{(c)}{\leq} \|\mathcal{T}(\ubf)\|\|\ubf-\vbf\| + \|\vbf\|\|\mathcal{T}(\ubf)-\mathcal{T}(\vbf)\|\\
&\overset{(d)}{\leq} \left(\Lambda^{-1}nB^2 + \left(\Lambda^{-1}+2\right)\sqrt{\nd}M\right)\|\ubf - \vbf\|, 
\end{align*}
where step $(a)$ follows by $\|\gbf\|\leq 1$ as established before, step $(b)$ follows by triangle inequality, step $(c)$ follows by Cauchy--Schwarz, and step $(d)$ follows by
\begin{align*}
\|\mathcal{T}(\ubf)\| &= \|\Lambda^{-1}\gbf^T\ubf\ubf + \Lambda^{-1}\ubf - \diag{\gbf}\ubf\|\\
&\leq \Lambda^{-1}\|\gbf\|\|\ubf\|^2 + \Lambda^{-1}\|\ubf\| + \hat{g} \|\ubf\|\nonumber\\
&\leq \Lambda^{-1}\nb M^2 + \Lambda^{-1}\sqrt{\nb}M + \sqrt{\nb}M.
\end{align*}
Similarly, we can show that
\begin{equation*}
T_3\leq \left(\Lambda^{-1}\nb M^2 + \left(\Lambda^{-1} + 2\right)\sqrt{\nb}M\right)\|\ubf-\vbf\|.
\end{equation*}
Proposition \ref{prop_converge} is then obtained by combining the upper bounds for $T_1$ through $T_6$.

\section*{Acknowledgment}
The authors would like to thank Yue M. Lu for suggesting the initialization scheme in~\eqref{eq:init_lam1}.
Computing resources provided by Boston University's Research Computing Services are also gratefully appreciated.

\bibliographystyle{IEEEtran}
\bibliography{ref}
\end{document}